\documentclass[journal,10pt]{IEEEtran}
\usepackage{subfigure}
\usepackage{moreverb}
\usepackage{epsfig}
\usepackage{amsmath,amssymb,amsthm,mathrsfs,amsfonts,dsfont}
\usepackage{adjustbox,lipsum}
\usepackage{algorithm,algorithmic}
\usepackage{amsfonts}
\usepackage{epsfig}
\usepackage{amssymb}
\usepackage{cleveref}
\usepackage{amsmath}
\usepackage{amsthm}
\usepackage{subfigure}
\usepackage{multirow}
\usepackage{rotating}
\usepackage{graphicx}
\usepackage{tabularx}
\usepackage{array}
\usepackage{anyfontsize}
\usepackage{color,soul}
\usepackage{graphicx,dblfloatfix}
\usepackage{epstopdf}
\usepackage{blindtext}
\usepackage{amsmath}
\usepackage{amsthm,amssymb,amsmath,bm}
\usepackage{subfigure}
\usepackage{amsfonts}
\usepackage{epsfig}
\usepackage{amssymb}
\usepackage{amsmath}
\usepackage{cite}
\usepackage{graphicx}
\usepackage{fancyhdr}
\usepackage{subfigure}
\usepackage[subfigure]{tocloft}
\usepackage[font={small}]{caption}
\usepackage{subfigure}
\usepackage{tabularx}
\usepackage{amsthm}
\usepackage{cite}
\usepackage{amsthm}
\newtheorem{theorem}{\textbf{Theorem}}
\newtheorem{lemma}{\textbf{Lemma}}

\begin{document}
\title{
 Robust Resource Allocation for PD-NOMA-Based MISO Heterogeneous Networks with CoMP Technology
}

\author{Atefeh Rezaei, Paeiz Azmi, Nader Mokari and Mohammad Reza Javan
\thanks{Atefeh Rezaei, Paeiz Azmi and Nader~Mokari are with the Department of Electrical and Computer Engineering, Tarbiat
Modares University, Tehran, Iran. Mohammad~R.~Javan is with the Department of Electrical and Robotics Engineering, Shahrood University, Shahrood.}
}

\maketitle
\begin{abstract}
	In this paper, we consider a hybrid scheme of coordinated multi-point (CoMP) technology in MISO heterogeneous communication networks based on power domain non-orthogonal multiple access (PD-NOMA).  We propose two novel methods based on matching game with externalities and also successive convex approximation (SCA) to realize the hybrid scheme where the number of the cooperative nodes is variable.
	 Moreover, we propose a new matching utility function to manage the interference caused by CoMP and NOMA techniques.
 We also devise robust beamforming to cope with the channel uncertainty. In this regard, we focus on both worst-case and stochastic cases in imperfect CSI information to increase the achievable data rate.
  We provide the complexity analysis of both schemes which shows that the complexity of the stochastic approach is more than that of the worst-case method.
 Results evaluate the performance and the sensibility of our proposed methods.
\end{abstract}

 \section{INTRODUCTION}
\IEEEPARstart{T}{he} idea of heterogeneous networks (HetNets) is to bring the network access point closer to the user by which the performance of resource usage and communications is improved. HetNets consist of several types of base stations (BSs), e.g.,  macrocells and femtocells, with different capabilities, transmit powers, and coverages \cite{wang2013het}. In such a network, due to spectrum reuse and dense utilization of BSs, intercell interference is of most concern. Coordinated multipoint (CoMP) is one of the promising techniques to alleviate the effect of intercell interference. Joint Processing (JP), which can be fell into Joint Transmission (JT-CoMP) and Dynamic Point Selection (DPS), is one variety of CoMP in which data can be simultaneously available at all of the cooperative nodes. Nevertheless, this scheme has many challenges such as high signaling and latency \cite{dps} and \cite{irmer}. Since the cooperation process in high order may not be required, we offer a method which allows cooperative nodes be variable based on the conditions.
To increase the spectral efficiency, the power domain non-orthogonal multiple access (PD-NOMA) scheme is presented in which the same spectrum is shared among several users (called NOMA set).
 Each user performs successive interference cancellation (SIC) to cancel the interference from other users. Devising an efficient resource allocation schemes for PD-NOMA-based transmission could be very challenging due to the fact that determining the NOMA set, transmit power of users in the NOMA set, and the SIC ordering would be very complicated in real scenarios.
 
The performance of the resource allocation method used in the network is highly depends on the availability of the channel state information at the transmitter (CSIT). However, in most practical scenarios, such a perfect knowledge is not available due to the limited capacity of the feedback channels, high signaling, and estimation errors \cite{wang2013het,goldsmith}. In such cases, the network should adopt robust methods to cope with the imperfectness of CSIT. 

Motivated by the above mentioned facts, we aim to devise an efficient robust resource allocation scheme for MISO networks based on the CoMP and the PD-NOMA. In such scheme, finding the NOMA set, the set of BSs performing CoMP for users, transmit power variables, and robustness against channel certainty is attained by formulating the resource allocation problem into an optimization problem and solving the resulting optimization problem via efficient iterative algorithms.

\subsection{Related Work}
 
 In \cite{atefeh2016}, we consider a   homogeneous MIMO network which works based on JT scheme when uncertainty of CSIT is taken into account. In order to solve this kind of non-convex problems, we employ some approximations like Bernstein inequality and semideﬁnite relaxation (SDR).
The authors in \cite{newref2}, investigate a TDD based MIMO network considering coordinated transmission. In order to minimize the power consumption, they perform RRH activation and robustness by using the group sparse
         beamforming method and also the Bernstein type inequality.
  The authors in \cite{khaled2015}   consider a single-carrier   network and propose the stochastic difference of convex programming (SDC) algorithm.
  	In \cite{slema}, worst-case optimization constraint is rewritten as a linear matrix inequalities by the S-procedure method. 
  	 In \cite{wang2013het}, a single-carrier HetNet without any cooperation process is considered and robustness solutions in no CSI and partial CSI feedback are proposed which are based on Bernstein inequality and SDR method. 
  	 In \cite{newref}, a time-division duplex (TDD) based HetNet with hybrid analog design for MBS are proposed. To find digital beamforming vectors, a power minimization problem with outage probability constraints which are approximated by the Bernstein-type inequality is solved.
       In \cite{mlft}, a downlink dual connectivity mode of a PD-NOMA-based heterogeneous cooperative network is assumed which it aimes to provide an energy efficient system.
The authors in \cite{mlft} employ a successive convex approximation (SCA) approach with Dinkelbach algorithm.
However, in the system model of \cite{mlft}, it is assumed that the transmitter is single antenna and CSI is assumed to be perfectly known, and thus robustness challenge is not addressed. \\
Recently, matching theory has attracted a lot of attention to solve the optimization problems in 5G networks. Most of the existing methods for multi-dimensional matching problems fall into two categories \cite{2dim1}: 1) transform the multi-dimensional matching
problem into the two-dimensional matching problems as to the pairing algorithm in \cite{2dim2}; 2) construct the
hypergraph model \cite{hyper1} or k-set packing problem \cite{hyper2}.
Although there are some works which employ the matching theory, none of them investigated cooperative NOMA based network and they ignored externalities and their framework are traditional.
 In \cite{nomamatch}, a matching-theory-based user scheduling and the optimal sensing duration adaptation are proposed in an alternate iteration framework for a cognitive OFDM-NOMA
systems where externalities are ignored in the matching algorithm.  In \cite{nomamatch1}, a greed
bidirection subchannel matching scheme  without externalities is provided for
NOMA system by selecting the users who have the
maximum subchannel energy efficiency.
In \cite{Shiva}, the matching theory is used to manage the  co/cross-tier interferences between D2D and cellular communications caused by resource sharing. Hence, the matching is an effectual tool to manage  all interferences in a network and we apply this idea in our proposed framework.

  The authors in \cite{maingame},  consider an OFDMA network in the uplink transmissions case. They  employ a one-to-many matching
game theory algorithm for user association and a one-to-one matching game for channel allocation problems. Moreover, the transmission process is assumed in the  non-cooperative mode where the CSI is perfectly known. 
In this paper, we propose a new algorithm for hybrid cooperative node association via  many-to-many matching
game and sub-channel allocation in PD-NOMA-based MISO system via  many-to-one matching
game based on the general framework of  \cite{maingame}. Unlike the methods in \cite{maingame}, there are some extra-interferences due to the NOMA and  CoMP in our considered network. Moreover, we assume externalities in our matching game to insure stability of our proposed method.\\
      The authors in \cite{Li} consider a distributed network which is based on single carrier CS/CB scheme. They assume that each BS knows  the perfect CSI of all
            UEs as local CSI while only the CSI from
            the other BS is not available.
             It is proposed that
            the BSs are able to zero-force the interference.
            However, the robust beamforming in imperfect CSIT challenge and specially probabilistic robustness which is generally intractable is not addressed in \cite{Li}.  
            In Table \ref{table5}, we compare some of the works  based on the CoMP technology from the perspective of robustness and resource allocation with the considered  multiple access technology.
\begin{table*}
	\begin{center}
		\caption{Comparison of CoMP based works from the perspective of robustness and resource allocation with multiple access technology}\label{table5}
		\resizebox{17.5cm}{!}{
			\begin{tabular}{|c|c|c|c|c|c|c|c|}
				\hline 
		\textbf{\begin{tabular}{@{}c@{}}Refe-\\rences \end{tabular}} & \textbf{\begin{tabular}{@{}c@{}}Multiple Access \\ Technology \end{tabular}} & \textbf{Infrastructure} & \textbf{Variables} & \textbf{\begin{tabular}{@{}c@{}}Objective \\ Function \end{tabular}}&\textbf{\begin{tabular}{@{}c@{}c@{}}QoS and\\ Constraints \end{tabular}} &\textbf{\begin{tabular}{@{}c@{}c@{}}Robustness Strategy\\ (Imperfect CSI)  \end{tabular}}& \textbf{\begin{tabular}{@{}c@{}c@{}}CoMP\\ Scheme \end{tabular}} \\ 
	\hline
	\cite{atefeh2016}& single-carrier & \begin{tabular}{@{}c@{}c@{}}Homogeneous\\ MISO 
	\end{tabular} & \begin{tabular}{@{}c@{}} antenna beam width\end{tabular}  & \begin{tabular}{@{}c@{}}Minimizing \\ power consumption\end{tabular}  & \begin{tabular}{@{}c@{}} Probabilistic \\Received SINR\end{tabular}& \begin{tabular}{@{}c@{}}Bernstein-type\\ inequality\end{tabular} & JT   \\ 
	\hline 
				\cite{khaled2015}& single-carrier & \begin{tabular}{@{}c@{}c@{}}Homogeneous\\ MISO 
				\end{tabular} & \begin{tabular}{@{}c@{}} antenna beam width\end{tabular}  & \begin{tabular}{@{}c@{}}Minimizing \\ power consumption\end{tabular}  & \begin{tabular}{@{}c@{}}Probabilistic Received SINR\end{tabular}& \begin{tabular}{@{}c@{}}SDC\\Algorithm\end{tabular} & JT   \\ 
				\hline 
		\cite{slema}& single-carrier & \begin{tabular}{@{}c@{}c@{}}Homogeneous\\ MISO 
		\end{tabular} & \begin{tabular}{@{}c@{}} antenna beam width\end{tabular}  & \begin{tabular}{@{}c@{}}Minimizing \\ power consumption\end{tabular}  & \begin{tabular}{@{}c@{}}Probabilistic Received SINR and \\ Interference limitation \end{tabular}& \begin{tabular}{@{}c@{}}
		S-procedure  \end{tabular} & JT   \\ 
		\hline 
				 \cite{newref2}& TDD & \begin{tabular}{@{}c@{}c@{}}Heterogeneous\\ MISO 
				\end{tabular} & \begin{tabular}{@{}c@{}} antenna beam width\end{tabular}  & \begin{tabular}{@{}c@{}}Minimizing \\ power consumption\end{tabular}  & \begin{tabular}{@{}c@{}}Probabilistic Received SINR and\\
				 Interference limitation
				\end{tabular}& \begin{tabular}{@{}c@{}}Bernstein-type \\inequality\end{tabular} & Hybrid   \\ 
		\hline  
		\cite{Li}& Single-carrier & \begin{tabular}{@{}c@{}c@{}}Homogeneous\\ MISO
		\end{tabular} & \begin{tabular}{@{}c@{}} antenna beam width\end{tabular}  & \begin{tabular}{@{}c@{}}Maximizing \\ throughput \end{tabular}  & \begin{tabular}{@{}c@{}}Achievable data rate and\\
		 Interference limitation
		\end{tabular}& \begin{tabular}{@{}c@{}}Non- \\robust\end{tabular} & CS/CB   \\
		\hline
				 \cite{mlft}& PD-NOMA & \begin{tabular}{@{}c@{}c@{}}Heterogeneous\\ SISO 
				\end{tabular} & \begin{tabular}{@{}c@{}} radio resource allocation
				 \end{tabular}  & \begin{tabular}{@{}c@{}}Maximizing \\ throughput\end{tabular}  & \begin{tabular}{@{}c@{}} achievable data rate 
				\end{tabular}& \begin{tabular}{@{}c@{}}Non-\\ robust\end{tabular} & Hybrid   \\ 
				\hline 
				\hline 
				\hline
				\textbf{Our Work}&\textbf{ \begin{tabular}{@{}c@{}} PD-NOMA\\\end{tabular}} & \textbf{\begin{tabular}{@{}c@{}} Heterogeneous  \\MISO\end{tabular}} & \textbf{\begin{tabular}{@{}c@{}c@{}}Joint radio resource\\ allocation and beamforming \end{tabular}} & \textbf{\begin{tabular}{@{}c@{}c@{}}Maximizing\\ throughput \end{tabular}} & \textbf{\begin{tabular}{@{}c@{}}Probabilistic achievable data rate\\ and interferance management \end{tabular}}& \textbf{\begin{tabular}{@{}c@{}} D.C. approximation with\\ Euclidean uncertainty set and \\Bernstein-type inequality \end{tabular}}& \textbf{Hybrid}  \\ 
				\hline
		\end{tabular}}
	\end{center}
\end{table*}
\subsection{Contributions}
Since the 5G network employes some techniques which lead to some design challenges as extra-interferences and unnecessary cooperations, we aim to model a practical and flexible network based on these techniques and attempt to propose some methods for use of advanced 5G networks.
As far as we
find out, there is no comprehensive work which consider the joint BS association and channel allocation in cooperative NOMA systems especially in the MIMO case with imperfect CSIT effect.
 
The main contributions and features of this paper can be summarized as follows:
 \begin{itemize}
 	\item 
 	\textbf{CoMP Scheme:} To eliminate unnecessary cooperations, we consider a hybrid scheme where  different antennas, as transmission nodes, necessarily are not in a fixed cooperative set (CS). In this regard, we propose a novel method based on matching games. Generally, unlike other existing works in CoMP design, we assume each antenna of a FBS can join to different CSs.

\item 
\textbf{Architecture of Network:}
B‌ased on our researches, all of the mentioned papers in the related works section which jointly investigate robustness of the CoMP and MIMO networks discussed a single-carrier network or they consider SISO multicarrier networks with perfect CSI. In this paper, we assume a cooperative network in multi-carrier conditions considering uncertainty of the CSIT. This problem has  a three dimensional  matching concept. Hence, we propose a new matching algorithm with a low complexity. Since multiplexing of
multiple users on the same frequency channel leads to co-channel interference (CCI), SIC must be performed at the receivers. In this regard, we introduce a new probabilistic SIC constraint which is strongly intractable.
 \item 
 	\textbf{ Advanced Interference Management:} In order to remove the extra-interference due to SIC and CoMP, we propose a novel approach to apply an advanced interference management method based on the matching utility functions. We pair the cooperative nodes in the CoMP set and users which are multiplexed on the same subcarriers in such a way to reduce the interference with harmful effect on other users. Moreover, to achieve a stable solution in practical networks, we consider externalities and employ swap-matching idea.
\item
 \textbf{Robustness Method:} We consider robustness in both stochastic and deterministic cases. The considered scenario is more intractable due to probabilistic SIC constraint. To solve the proposed robust optimization problem, a novel alternative sequential algorithm is proposed. Moreover, the convergences  of the iterative algorithms are proved and their computational complexities are
investigated.
\end{itemize}
\subsection{Organization}
The rest of this paper is organized as follows:
$  $
  The considered system model is presented in Section \ref{secsys}. The proposed resource allocation problem is formulated in Section \ref{problem}. In Section \ref{MATCHING}, a new matching game based solution is proposed. We investigate the convergence and the computational complexity of the proposed methods in Section \ref{seccomputation}. Simulation results are in Section \ref{secsimulation}, and the paper is concluded in Section \ref{secconclusion}.

Notations: We use $ \circ $ to define Hadamard product of two vectors while $ * $ represents the common matrix multiplication. $ \langle \textbf{a},\textbf{b}\rangle $ represents inner product of two matrices and $\textbf{ A}\succeq 0 $ indicates that $ \textbf{A} $ is a positive semidefinite (PSD) matrix. In addition, $ \|\textbf{.}\|_F $ and $ \|\textbf{.}\| $ denote Frobenius norm of the matrix and Euclidean norm of a vector, respectively. Trace of a matrix is defined via $ \text{trace}\;[\textbf{A}] $. The conjugate transpose of a matrix $ \textbf{A} $ is denoted by $ \textbf{A}^H $. The complex
space of $n$-dimensional vectors is described using $ \mathbb{C}^{n}$. $ \lambda_{\text{max}}\:(.) $ denotes the maximum eigenvalue of matrix. $ \text{Re}\:\{.\} $ and $ \mathbb{E}\:\{.\} $ are the real part and mean of associated argument. The maximum number of linearly independent row vectors in the matrix is shown using $ \text{rank}\:(.)$. The expression $ \mathcal{A} \subset \mathcal{B} $ defines $\mathcal{A}$ as a subset of the set $\mathcal{B}$, and also
$\cup $ denotes the union of two sets.
$\{\mathcal{A}\backslash (a)\}$ equals a subset contains all of the elements of set $\mathcal{A}$ except element $a$.
\section{SYSTEM MODEL}\label{secsys}
  We assume a heterogeneous network with one MBS and several FBSs while the spectrum of the MBS is shared with all of the FBSs which are operating inside the coverage of the corresponding MBS. Hence, the interference effect of macrocell is taken into account. Consequently, we consider a HetNet presented at Fig. \ref{fig:Drawing2} with $ F $ FBSs in a set as  $ \mathcal{F} =\{1, ... , F\} $ where each of them is  equipped with $T_{f}$ antennas and a MBS with $T_{m}$ antennas. 
We define the set of femtocell users (FUEs) as $ \mathcal{K} =\{1, ... , K\} $ and  the set of  available subcarriers in the FBSs as $  \mathcal{N} =\{1, ... , N\} $.
We assume that
$\hat{\textbf{h}}_{\mathcal{F}_{k,\: n}}^f \in \mathbb{C}^{T_f}$
 and $ \hat{\textbf{w}}_{k,\: n}^f \in \mathbb{C}^{T_f} $ display  channel coefficient and beamformer vector of the $f^{th}$ FBS to the $ k^{th} $ FUE at the $ n^{th} $ subcarrier.  In this manuscript, we just focus on determining the beamforming vector of FBSs to decreas the malicious effect of cooperative femtocells on both of the FUE and MUE. We suppose that $ \hat{\textbf{h}}_{\mathcal{MF}_{n}}^f \in \mathbb{C}^{T_f} $  displays the channel coefficients from the $f^{th}$  FBS  to the MUE over subcarrier  $ n $  and 
$ \textbf{h}_{\mathcal{FM}_{k,\: n}} \in \mathbb{C}^{T_m} $ are channel coefficients of  MBS  to the $ k^{th} $ FUE over  subcarrier $ n $. Further, we introduce $ \textbf{m}_{n} \in \mathbb{C}^{T_m} $ as beamformer vector of MBS at the $ n^{th} $ subcarrier.  In general case, we introduce 
\begin{align}
 \textbf{h}_{\mathcal{F}_{k,\: n}} =[(\hat{\textbf{h}}_{\mathcal{F}_{k,\: n}}^1)^{T},\: ..., \: (\hat{\textbf{h}}_{\mathcal{F}_{k,\: n}}^f)^T, \: ..., \: (\hat{\textbf{h}}_{\mathcal{F}_{k,\: n}}^F)^T]^T \in \mathbb{C}^{F T_f},\label{eq:ch1}
\end{align} 
as channel coefficients from all of the antennas as transmission nodes  to the $ k^{th} $ FUE over the $ n^{th} $ subcarrier. 
\begin{figure}
\centering
\includegraphics[width=0.85\linewidth]{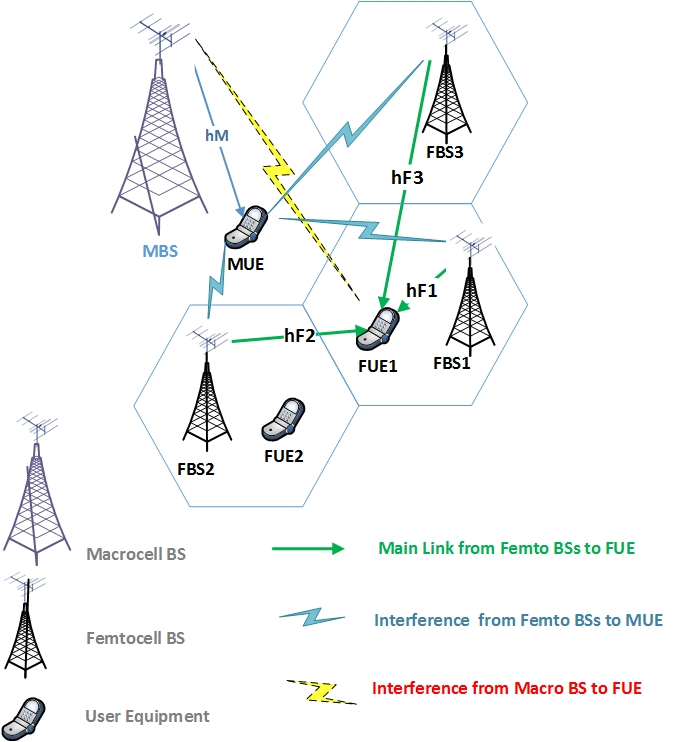}
\caption{HetNet Based on CoMP Technology}\label{fig:Drawing2}
\end{figure}
  In this network, we assume that all of the chosen coordinated  transmission nodes transfer data to the user over  subcarrier $ n $. It is noticeable that each element of this vector is  an antenna which is defined as an independent transmission node $a$. Therefore, 
in order to make macro diversity by the CoMP procedure, it is necessary that at least some parts of the femtocell's bandwidth be same and the user with critical condition can be served by the cooperate nodes on common subcarrier.
As multiple users can employ subcarrier $n$,  each
user $k$ adopts the SIC technique \cite{rev3}.
Consequently, SINR of FUE $ k$ on the $ n^{th} $ subcarrier is given by
\begin{align}
\Gamma_{(\textbf{w}_{k,\:n},\ \boldsymbol{\rho}_{k,\: n},\ \textbf{h}_{\mathcal{F}_{k,\: n}},\: \textbf{h}_{\mathcal{FM}_{k,\ n}})}=\frac{|\textbf{h}_{\mathcal{F}_{k,\: n}}^{H} (\textbf{w}_{k,\: n}\circ\boldsymbol{\rho}_{k,\: n})|^{2}}{I_{\mathcal{FM}_{k,\: n}}+I_{\mathcal{F}_{k,\: n}}+\sigma_{k,n}^{2}} ,\label{eq:1}
\end{align}
where $I_{\mathcal{FM}_{k,\: n}}=|\textbf{h}_{\mathcal{FM}_{k,\: n}}^{H} \textbf{m}_{n}|^{2}$ is the interference from MBS and $I_{\mathcal{F}_{k,\: n}}=\Sigma_{i\in \mathcal{K},\|\textbf{h}_{\mathcal{F}_{i,\: n}}\|> \|\textbf{h}_{\mathcal{F}_{k,\: n}} \|}\\|\textbf{h}_{\mathcal{F}_{k,\: n}}^{H} (\textbf{w}_{i,\: n}\circ\boldsymbol{\rho}_{i,\: n})|^{2}$ expresses the summation of inter-cell interference of other FBSs and intra-cell interference due to the PD-NOMA approach. 
We define $ \boldsymbol{\rho}_{k,\: n}=[(\hat{\boldsymbol{\rho}}_{k,\: n}^1)^T, ...,  (\hat{\boldsymbol{\rho}}_{k,\: n}^F)^T]^T=[\rho_{k,\: n}^{a}] \in \mathbb{Z}^{F T_f}$ as a multi-dimensional matrix where $ \rho_{k,\: n}^{a} $ is an integer variable. If  node $ a $ can transmit data of user $ k $ over subcarrier $ n $, $ \rho_{k,\: n}^{a}=1 $, and thus, if transmission for user $ k $ can not be performed over sucarrier $ n $ through node $ a $ for any reason, $ \rho_{k,\: n}^{a}=0 $.
$ \textbf{w}_{k,\: n}=[(\hat{\textbf{w}}_{k,\: n}^1)^T, ..., (\hat{\textbf{w}}_{k,\: n}^F)^T]^T\in \mathbb{C}^{F T_f} $ and $ \sigma_{k,n}^2 $ is the noise power. $ \hat{\boldsymbol{\rho}}_{k,\: n}^f
\in \mathbb{C}^{T_f}$  is the subcarrier indicator vector of the $f^{th}$ femtocell.  Accordingly, the achievable data rate at user $ k $ over subcarrier $ n $ is formulated by
\begin{align}
  r_{k,n}=\log_2(1+\Gamma_{(\textbf{w}_{k,\:n},\ \boldsymbol{\rho}_{k,\: n},\ \textbf{h}_{\mathcal{F}_{k,\: n}},\: \textbf{h}_{\mathcal{FM}_{k,\ n}})}) \label{rate} .
 \end{align}
\subsection{Channel State Information and Robustness}					
 Typically, by transmitting pilot symbols in the downlink transmission and estimating channels at the receiver sides, the estimated CSI can be feedbacked from the receiver to FBSs. 
 In this paper, we assume that the feedback links from the receiver to the transmitter are error-free and the estimated CSI at receiver is imperfect.
 To model the uncertainty of CSI, we choose additive
  error model as follows:
  \begin{align}\label{additive}
  \textbf{h}_{\mathcal{F}_{k,\: n}}=\bar{\textbf{h}}_{\mathcal{F}_{k,\: n}} + \textbf{e}_{\mathcal{F}_{k,\: n}},
  \end{align}
  where the notations $ \bar{\textbf{h}}_{\mathcal{F}_{k,\: n}} $ and $ \textbf{e}_{\mathcal{F}_{k,\: n}} $ denote the estimated channel coefficients at the receiver and the error vector  \cite{csidefinition,Caire}.
   We consider two schemes of CSI imperfection as follows: 
      \begin{itemize}
     \item In the worst-case, we assume $ \textbf{e}_{\mathcal{F}_{k,\: n}} $ is a norm bounded vector for analytical convenience. In this regard, we consider
    the Euclidean ball-shaped uncertainty set as follows:
    \begin{align}\label{ellipsoid}
     \mathcal{H}=\{\textbf{h}_{\mathcal{F}_{k,\: n}}:\textbf{h}_{\mathcal{F}_{k,\: n}}=\bar{\textbf{h}}_{\mathcal{F}_{k,\: n}}+\bar{\textbf{e}}_{\mathcal{F}_{k,\: n}},\| \bar{\textbf{e}}_{\mathcal{F}_{k,\: n}}\| \leq \zeta_{k,n}\}, 
    \end{align}
    where $ \zeta_{k,n} $ defines the error bound on the uncertainty region.
     $ \bar{\textbf{h}}_{\mathcal{F}_{k,\: n}} $ is defined as  the estimation coefficients of channels which is the center of the assumed ball \cite{atefeh2016},  \cite{csidefinition}, \cite{gharavol}. \item In the stochastic imperfect case, we assume slow fading channels while the error vectors follow the complex Gaussian distribution as \cite{atefeh2016}, \cite{bernstein2},  \cite{lowcomplexity}. Hence, $ \textbf{e}_{\mathcal{F}_{k,\: n}} $ is a complex Gaussian random vector with specific fixed mean and covariance matrix as
    \begin{align}\label{bbb}
    \textbf{e}_{\mathcal{F}_{k,\: n}}=\textbf{C}^{1/2}_{e,\mathcal{F}_{k,n}}\textbf{v}_{\mathcal{F}_{k,\: n}},
    \end{align}
    where $\textbf{C}_{e,\mathcal{F}_{f,k,n}}\succeq 0 $ is the covariance matrix of $ \textbf{e}_{\mathcal{F}_{k,\: n}} $, $ \textbf{v}_{\mathcal{F}_{k,\: n}} $ is the complex Gaussian random vector, i.e., $ \textbf{v}_{\mathcal{F}_{k,\: n}} \sim \mathcal{N} (\textbf{0},\textbf{I})$, and $ \textbf{I} $ is an identity matrix.
     \end{itemize}

 \begin{table*}
   		\caption{NOTATIONS}
   	\label{table-00}
 \begin{center}
 \resizebox{15cm}{!}{
 \renewcommand{\arraystretch}{2}
  	\begin{tabular}{| c| l|l|l|}
  	
  		\hline 
  		Notation& Description& $\bar{\textbf{h}}_{\mathcal{F}_{k,\: n}}$ & Estimation of $  \textbf{h}_{\mathcal{F}_{k,\: n}} $    \\
  		\hline
  		
  	  $\hat{\textbf{h}}_{\mathcal{F}_{k,\: n}}^f$& Channel vector from FBS $f$ to user $ k $ over subcarrier $ n $& $\bar{\textbf{h}}_{\mathcal{MF}_{ n}}$ & Estimation of $  \textbf{h}_{\mathcal{MF}_{ n}} $ \\
  	  \hline 
  	  
  	 $\hat{\textbf{h}}_{\mathcal{MF}_{n}}^f$& 	Channel vector from FBS $f$ to MUE over subcarrier $ n $&$\bar{\textbf{h}}_{\mathcal{FM}_{k,\: n}}$ & Estimation of $  \textbf{h}_{\mathcal{FM}_{k,\: n}} $\\\hline
  		 $\textbf{h}_{\mathcal{FM}_{k,\: n}}$& Channel vector from MBS to user $ k $ over subcarrier $ n $&$\textbf{e}_{k,n}$ & Error vector on $ \textbf{h}_{\mathcal{F}_{k,\: n}} $\\\hline
  		$ q_{\text{max}}$ & Maximum number of FUEs with common subcarrier $n$  & $ \textbf{z}_{k,n} $ & Error vector on $ \textbf{h}_{\mathcal{FM}_{k,\: n}} $ \\\hline
  		$ \epsilon$ & Stopping criterion accuracy for the CTNSA Algorithm &$\textbf{q}_{k,n}$ & Error vector on $ \textbf{h}_{\mathcal{MF}_{k,\: n}} $\\\hline
 	
  		 $\hat{\textbf{w}}_{k,\: n}^f$& Beamformer vector of FBS $f$ for user $ k $ over subcarrier $ n $ &$\zeta_k$ & Error bound related to $ \textbf{e}_{k,n} $ \\\hline
  		 	 
  		 		 $\textbf{m}_{\mathcal{FM}_{k,\: n}}$&  Beamformer vector of MBS for user $ k $ over subcarrier $ n $&$\kappa_k$ & Error bound related to $ \textbf{z}_{k,n} $ \\\hline
  		 	$ \varUpsilon_{CA}$ & Weighting parameter in CA phase&$\eta_k$ & Error bound related to $ \textbf{q}_{k,n} $ \\\hline
  		 		$\hat{\boldsymbol{\rho}}_{k,\: n}^f$& Subcarrier indicator  vector of FBS $f$ for user $ k $ on subcarrier $ n $ &$\textbf{C}_{e,\mathcal{F}_{f,k,n}}$ & Covariance matrix of $ \textbf{e}_{\mathcal{F}_{k,\: n}} $\\\hline
  		 		 		 	$ c_{\mathcal{FM}}^n$ & Cost unit of interference at MUE &$ R_{k} $ & Target rate at user $ k $\\\hline
  	$\beta, \alpha $ & Maximum tolerable outage
  	  	  		 for $ r_{k,n} $ and interference & $ \mathcal{K}  $ & Set of total FUEs  \\\hline
  		$\mathcal{H}_{\mathcal{F}_{k,\: n}}$ & Channel uncertainty set of $ \textbf{h}_{\mathcal{F}_{k,\: n}} $   &$  \phi $ & Utility function  \\\hline
  		$\mathcal{H}_{\mathcal{MF}_{ n}}$ & Channel uncertainty set of $ \textbf{h}_{\mathcal{MF}_{ n}} $   &	$\epsilon_M$ & Maximum interference power\\\hline
  		$\mathcal{H}_{\mathcal{FM}_{ k, n}}$ & Channel uncertainty set of $ \textbf{h}_{\mathcal{FM}_{k,\: n}} $  &$P_\text{max}$ & Maximum transmit power of FBS \\\hline
  		 $ \mathcal{A}_k   $ &  Set of transmission nodes assigned to FUE $k$
  		& $\mu$ & Mapping function \\\hline
             $  T_f $ &  Number of antennas in FBS
  		& $ T_m$ &  Number of antennas in MBS \\\hline
$ \boldsymbol{\chi}  $ &  Transmission node selection indicator matrix
  		& $  \mathcal{N} $ &  Set of FBSs in the network\\\hline
$F_\text{max}$ & Maximum number of cooperative FBSs
  		& $ \mathcal{N}$ &  Set of total available channels \\\hline
 $ \boldsymbol{\nu} $ & Channel allocation indicator matrix
  		& $\mathcal{A} $ &  Set of all the transmission nodes \\\hline
$  \mathcal{N}_a $ &  Set of available subcarrier at node $a$
  		& $ \mathcal{K}_a $ & Set of FUEs associated to $a$ \\\hline
$ \hat{N}_a  $ & Maximum FUEs serving by each antenna 
  		& $\textbf{h}_{\mathcal{F}_{k,\: n}}$ & $[\hat{\textbf{h}}_{\mathcal{F}_{k,\: n}}^f]$  \\\hline
$r_{k,n} $ & Data rate at user $ k $ on subcarrier $ n $
  		& $\textbf{h}_{\mathcal{MF}_{ n}}$ &$ [ \hat{\textbf{h}}_{\mathcal{MF}_{n}}^f] $  \\\hline
$  \varUpsilon_{CS} $ & Weighting parameter in CS phase 
  		& $\textbf{w}_{k,\: n}$ & $ [\hat{\textbf{w}}_{k,\: n}^f] $    \\\hline
$  c_i^n $ &  Cost unit of interference at FUE $i$ 
  		& $\boldsymbol{\rho}_{k,\: n}$ & $ [\hat{\boldsymbol{\rho}}_{k,\: n}^f] $ \\\hline
  	\end{tabular}
  	}
  \end{center}
  \end{table*}
\section{PROBLEM FORMULATION}\label{problem}
The beamformer vector of transmitters must be designed based on the channel models such that guarantees the outage occurs below a small predetermined probability threshold as follows:
1) We define achievable data rate based on \eqref{eq:1}  considering uncertainty sets similar to \eqref{ellipsoid} in a worst-case approach as
$ r_{(\bar{\textbf{h}}_{\mathcal{F}_{k,\: n}},\textbf{W}_{k,\:n},\ \boldsymbol{\rho}_{k,\: n})} \geq R_{k} $
where channel coefficients must be considered in uncertainty set $ \mathcal{H} $ and $ R_k $ is the minimum required achievable data rate.
 2) In probabilistic case, we define achievable data rate  considering imperfect CSI as
\begin{align}
 \Pr {\bigg \{ } \log_2(1+\Gamma_{(\textbf{W}_{k,\:n},\ \boldsymbol{\rho}_{k,\: n})} )\geq R_{k} {\bigg \} } \geq 1-\beta,\label{eq:2b}
\end{align}
where $ R_{k}$ and $ \beta $ respectively denote the target rate for FUE $ k $ and the maximum tolerable outage
probabilities.  In downlink transmission of NOMA based systems, each user equipment can successfully detect and cancel the interference from all users with lower channel gains whereas the interference of users with higher channel gains must be considered at the desired signal. 
 In order to ensure successive SIC procedure at each user, there is an information theory constraint \cite{rev3,cui2016} as follows: 
\begin{align}
 &\Pr {\bigg \{ } \Gamma_{(\textbf{w}_{k,\:n},\ \boldsymbol{\rho}_{k,\: n},\ \textbf{h}_{\mathcal{F}_{j,\: n}},\: \textbf{h}_{\mathcal{FM}_{j,\:n}})} \geq \nonumber \\
  &
  \Gamma_{(\textbf{w}_{k,\:n},\ \boldsymbol{\rho}_{k,\: n},\ \textbf{h}_{\mathcal{F}_{k,\: n}},\: \textbf{h}_{\mathcal{FM}_{k,\:n}})} {\bigg \} }\nonumber \\
 &  \geq 1-\beta \:\:\: \forall j,k\in \mathcal{K}, \|\textbf{h}_{\mathcal{F}_{j,\: n}}\|> \|\textbf{h}_{\mathcal{F}_{k,\: n}} \| .\label{sic}
\end{align}
  This constraint expresses the SINR of user $ k $ at users with higher channel qualities must be higher than the SINR of user $ k $ at itself.
In the more general concept, we assume that CS which includes different antennas in various FBSs is variable and not necessarily all of these nodes are in the CS. Therefore, the nodes related to a CS could be determined by solving an optimization problem. This problem has three dimensional matching concept which is hard to solve.   
 To address this problem, we define the antennas as transmission nodes in a new set as $\mathcal{A}=\{1,...,F*T_f\}$. To solve this problem,  we transform it into two two-dimensional matching problems. Hence, we define some new integer variables as $\chi_{k,a} $ and $\nu_{k,n}$ so that $\rho_{k,\: n}^{a}= \chi_{k,a}*\nu_{k,n}$. The variable $\chi_{k,a} $ denotes that node $a$ is assigned  to FUE $k$ or not, and $\nu_{k,n}=1$ denotes the set of cooperative nodes can use subcarrier $n$ to FUE $k$. 
 The optimization problem for maximizing the total throughput of the considered network is described as follows:
\begin{subequations} \label{eq:2} 
     \begin{align}
      & \underset{\textbf{W} , \boldsymbol{\chi}, \boldsymbol{\nu}}{\text{max}} \quad  \Sigma_{n\in \mathcal{N}}\Sigma_{k\in \mathcal{K}}r_{k,n},  \label{eq:2a}  \\
      &\text{s.t.}\hspace{0.3cm}\eqref{eq:2b}, \eqref{sic},  \nonumber \\
              & \hspace{1cm} \Pr {\bigg \{ } \Sigma_{k\in \mathcal{K}}  |\textbf{h}_{\mathcal{MF}_{n}}^{H}(\textbf{w}_{k,\: n}\circ\boldsymbol{\rho}_{k,\: n})|^2 \leq \epsilon_M {\bigg \} } \geq 1-\alpha, \label{eq:2c} \\
                 &\hspace{1cm} 0 < \Sigma_{n\in \mathcal{N}}\Sigma_{k\in \mathcal{K}}\|(\hat{\textbf{w}}^f_{k,\: n}\circ\hat{\boldsymbol{\rho}}^f_{k,\: n})\|^{2}\leq P_\text{max}, \forall  f\in \mathcal{F} \label{eq:2d} \\
 &\hspace{1cm} \Sigma_{a\in \mathcal{A}}\chi_{k,a}\leq F_\text{max}, \forall k\in \mathcal{K}, \label{eq:2e1} \\
&\hspace{1cm} \Sigma_{k\in \mathcal{K}}\chi_{k,a}\leq \hat{N}_a, \forall a\in \mathcal{A}, \label{eq:2e2} \\
              &\hspace{1cm} \Sigma_{k\in \mathcal{K}}\nu_{k,n}\leq q_\text{max}, \forall n\in \mathcal{N}, \label{eq:2e3} \\
&\hspace{1cm} \Sigma_{n\in \mathcal{N}}\nu_{k,n}\leq 1, \forall k\in \mathcal{K}, \label{eq:2e4} \\
             & \hspace{1cm} \chi_{k,a},\:\nu_{k,n}\in
             \{0, 1\},
                             \forall k\in \mathcal{K}, n\in \mathcal{N}, a\in \mathcal{A},                 
              \label{eq:2f} 
     \end{align}
\end{subequations} 
where all of the beamformer vectors of FBSs are considered in a matrix as \textbf{W}, and also matrix $ \boldsymbol{\rho} $ defines a three dimensional matrix for all subcarrier indicators so that $\boldsymbol{\rho}=\boldsymbol{\chi}*\boldsymbol{\nu}$.
Practically, there are some extra-effects of FBSs on the macrocell \cite{wang2013het}.
In the CoMP network,  these effects are not ignorable.
 It is intelligent to assume an interference power constraint
on the MUE.
 In this regard, 
 \eqref{eq:2c} is considered to improve the overall performance of the network where $\epsilon_M$  is the preset target value and $\alpha$ denotes the maximum tolerable outage
probabilities for interference power constraints. 
Due to the femtocell hardware limitations, the number of FUEs associated to each
FBS is restricted \cite{femto8}. We limit this number of UEs with \eqref{eq:2e2} where $\hat{N}_a$ is the maximum number of UEs that can be served by each antenna as in \cite{maingame}.
 Practically, the transmit power of each FBS is a function of beamforming coefficient of antenna $ a $ when subcarrier $ n $ is assigned to the intended antenna which is in the serving cooperative set of user $ k $. Note also that we assume the transmit signal power equals to 1.
  Hence, we assume constraint ($ 9c $) as a limitation on the discriptive equivalent of the total transmit power of each FBS, where $P_\text{max}$ is defined as an upper bound.
 To manage signaling volume and delays, we introduce a complexity constraint as \eqref{eq:2e1} which limits the total number of coordinated transmission nodes with the maximum order of cooperation $F_{\text{max}}$. In \eqref{eq:2e3} and \eqref{eq:2e4}, we assume that spectrum can be shared between MBS and FBSs while each subcarrier can be assigned to a set of users and every user can be served  by one subcarrier and $q_{\text{max}}$ denotes the maximum number of interfered users through the PD-NOMA technique.
\section{MATCHING GAME BASED RESOURCE ALLOCATION}\label{MATCHING}
Optimization problem described in \eqref{eq:2} is a probabilistic mixed-integer nonlinear programming (MINLP) problem. Generally, we tend to perform a joint transmission nodes assignment and subcarrier allocation with beamforming. To solve \eqref{eq:2}, we choose an iterative based framework which has three independent phases:  1)  Given an initialized beamforming vectors, we propose a new algorithm based on many-to-many matching game which perform user association. The output of this phase introduces as $\boldsymbol{\chi}^{*}$.   
2) Given $\boldsymbol{\chi}^{*}$ and $\boldsymbol{W}^{*}$, the subcarrier assignment problem can be solved by a  many-to-one matching game algorithm where the output of this phase introduces $\boldsymbol{\nu}^{*}$.   
3) At third phase, beamforming design based on the proposed association at previous phases is performed.
These phases  are sequentially applied until the problem converges to a feasible solution of $\{\boldsymbol{\chi}^{*}, \boldsymbol{\nu}^{*},\boldsymbol{W}^{*} \}$.

\subsection{Cooperative Transmission Node Selection (CTNSA) Algorithm}

Since the clustering is an important issue in the MIMO-5G networks, we investigate it in a general problem as follows:
\begin{subequations} \label{CTNSA} 
     \begin{align}
& \mathcal{P}_{\text{CTNSA}}: \hspace{0.25cm}  \underset{ \boldsymbol{\chi}}{\text{max}} \quad  \Sigma_{n\in \mathcal{N}}\Sigma_{k\in \mathcal{K}}r_{k,n},  \label{eq:2aa}  \\
      & \hspace{1.9cm} \text{s.t.}\hspace{0.3cm}\eqref{eq:2e1}, \eqref{eq:2e2}, \eqref{eq:2f}. \nonumber         
     \end{align}
\end{subequations} 
 As we mentioned before, we use 
\eqref{eq:2e1} and \eqref{eq:2e2} for this algorithm which works based on many-to-many matching concept. 
By the definition of the different antennas as transmission nodes, $\chi_{k,a}$ defines node $a$ is joined to the cooperative set of user $k$ or not. In this problem, we describe a CS with $\mathcal{A}_k\subset \mathcal{A}$ as the set of transmission nodes assigned to FUE $k$ and also  $ \mathcal{K}_a \subset \mathcal{K} $ as the set of users associated to node $a$. Additionally,
 we define $\mathcal{N}_a \subset \mathcal{N}$ as the set of available subcarriers at node $a$ while the set of FUEs multiplexing on subcarrier $n$ is denoted by $ \mathcal{C}_n \subset \mathcal{K} $. 
\subsubsection{Definition of a matching function}
Optimization problem $ \mathcal{P}_{\text{CTNSA}}$  can be defined by a tuple ($\mathcal{A}, \mathcal{K},  \succ_{\mathcal{A},CS}, \succ_{\mathcal{K},CS}$). Here, $\succ_{\mathcal{A},CS}$ and $\succ_{\mathcal{K},CS}$ denote the sets of the preference relations of FUEs and transmission nodes, respectively.  We define two disjoint finite sets of players $\mathcal{A}$ and $\mathcal{K}$, and also a mapping function $\mu_{CS}$  such that:
1) $a\in \mu_{CS}(k) \longleftrightarrow k\in \mu_{CS}(a)$; 
2) $|\mu_{CS}(k)|\leq F_{\text{max}}$ and $|\mu_{CS}(a)|\leq \hat{N}_a $.
Instead of \eqref{eq:2b} and \eqref{eq:2c}, we introduce some utility functions to each component of players until they construct their preference list in a decreasing order. We use the average received SINR over all subchannels as utility function of the FUE which is the
most common criterion for user association \cite{maingame}, \cite{uacriterion} as follows:
\begin{equation} \label{ukcs} 
     \begin{aligned} 
& \varphi_{CS}^k(a)=\log_2(1+\Sigma_{n\in \mathcal{N}_a} \gamma_{k,n}^a).
     \end{aligned}
\end{equation} 
In this case, $\gamma_{k,n}^a$ describes independent effect of each node $a$ on the received SINR as
\begin{align}\label{isinra} 
\gamma_{k,n}^a=\frac{\chi_{k,a}\nu_{k,n}|\bar{h}_{\mathcal{F}_{k,\: n}}^{a} \hat{w}_{k,\: n}^a|^{2}}{I_{\mathcal{FM}_{k,\: n}}+I_{\mathcal{F}_{k,\: n}}+\sigma_{k,n}^{2}} , \forall k\in \mathcal{K}_a, n\in \mathcal{N}_a, 
\end{align}
where the interference\\ $I_{\mathcal{F}_{k,\: n}}=\Sigma_{i\in \mathcal{K},|\bar{h}_{\mathcal{F}_{i,\: n}}^a|^2> |\bar{h}_{\mathcal{F}_{k,\: n}}^a|^2}|(\Sigma_{a\in \mathcal{A}}\chi_{i,a}\nu_{i,n}\bar{h}_{\mathcal{F}_{k,\: n}}^{a} \hat{w}_{i,\: n}^a)|^{2}$.  \\
In \cite{maingame}, a new function is introduced as a utility function for user association at uplink transmission of HetNets. we implement the same function for our problem as
 \begin{equation} \label{uacs} 
     \begin{aligned} 
& \varphi_{CS}^a(k)=\varUpsilon_{CS} \Sigma_{n\in \mathcal{N}}\frac{|\bar{h}_{\mathcal{F}_{k,n}}^a\hat{w}_{k,n}^a|^2}{2^{R_k}-1}-\varTheta_{k,a},
     \end{aligned}
\end{equation}
where $\varUpsilon_{CS}$ is a weighting parameter capturing the average direct channel gain from node $a$ to the FUE. 
Although the CoMP can improve the received signal of users, it causes significant interference in other receivers, i.e., MUE and FUE $i$ ($\forall i\neq k$) especially in the case of unsuccessful SIC procedure. 
 Therefore, we propose an advanced interference management through introducing
 $\varTheta_{k,a}$ which
quantifies the aggregated interference that node $a$
causes to the MUE and also the other FUE $i$ on
all subchannels which is defined as $\varTheta_{k,a}=\Sigma_{n\in \mathcal{N}_a}( \varTheta^n_{\mathcal{MF}_{k,a}}+ \varTheta^n_{\mathcal{F}_{k,a}})$ where
\begin{subequations}
     \begin{align}
&
\varTheta^n_{\mathcal{MF}_{k,a}}=c^n_{\mathcal{MF}}\varpi^n|\bar{h}^a_{\mathcal{MF}_{n}}\hat{w}_{k,n}^a|^2,
 \label{uacs1}\\ 
& \varTheta^n_{\mathcal{F}_{k,a}}=\Sigma_{i\in \mathcal{K}_a\backslash \{k\}, |\bar{h}^a_{\mathcal{F}_{i,\: n}}|\geq|\bar{h}^a_{\mathcal{F}_{k,\: n}}|}c_i^n|\bar{h}^a_{\mathcal{F}_{i,n}}\hat{w}_{k,n}^a|^2. \label{uacs2} 
     \end{align}
\end{subequations}
\eqref{uacs1} and \eqref{uacs2} decrease the interference of node $a$ on MUE and other FUEs, respectively.\\
 $\varpi^n=\text{max}\big(0,  (\Sigma_{i\in \mathcal{K}_a}|\bar{h}^a_{\mathcal{MF}_{n}}\hat{w}_{k,n}^a|^2-\epsilon_M)/\epsilon_M\big)$ is defined to quantify
the degree of violation of the constraint \eqref{eq:2c}. 
$c^n_{\mathcal{MF}} $ and $c_i^n$ are the costs per unit of
the interference power at the MUE and FUE $i$, respectively.
The proportions of 
$c^n_{\mathcal{MF}}$ and $c_i^n$ can be set, based on the priority of users.
For example,
 $c^n_{\mathcal{MF}}\gg c_i^n$ indicates the priority of MUEs and guarantees that solutions with harmful effect on MUEs can be blocked. 
\subsubsection{Description of the stopping criterion}
 In order to realize the  hybrid scheme, we introduce a new criterion in addition to \eqref{eq:2e1} which limits the size of CS as required. Accordingly, we define 
\begin{align}\label{isinr12} 
\gamma_{k,n}=\frac{|\Sigma_{a\in \mathcal{A}_k}\chi_{k,a}\nu_{k,n}\bar{h}_{\mathcal{F}_{k,\: n}}^{a} \hat{w}_{k,\: n}^a|^{2}}{I_{\mathcal{FM}_{k,\: n}}+I_{\mathcal{F}_{k,\: n}}+\sigma_{k,n}^{2}} , 
\end{align}
as the SINR when CS of FUE $k$ include node $a$. As well as, the average received SINR is defined $\varPhi_{k}=\log_2(1+\Sigma_{n\in \mathcal{N}}\gamma_{k,n})$. Now, we introduce a new criterion constraint which determines node $a$ can join to the CS or not. Actually, because of the hybrid scheme, the number of the cooperative nodes is variable on request. The mentioned criterion is defined as follows:
 \begin{align}\label{phicri} 
& D^{\varPhi_k}_{CS}=|\varPhi_{k}^{\{\mathcal{A}_k \cup\{a\}\}}- \varPhi_{k}^{\{\mathcal{A}_k \}}| \leq \epsilon.
\end{align}
\subsubsection{Structure of the CTNSA Algorithm}
 
 \begin{algorithm}
 \caption{Matching CTNSA Algorithm.} 
 \textbf{\textit{Step 1: Initialization:}} \\
1.\quad Preset $ \mathcal{L}^{req}_a=\o{}$, $ \mathcal{L}^{rej}_a=\o{}$, $F_{\text{max}}$, $\hat{N}_a$ \: $\forall k, a$. \\
 \textbf{\textit{Step 2: Utility Computation:}}\\
       2.\quad  construct $\mathcal{P}_{k,CS}$  using $\varphi_{CS}^k(a)$ $\forall k$.\\
 \textbf{\textit{Step 3: Find stable matching ($\mu_{CS}$) without externalities:}}\\
               3.\quad \textbf{while} $\Sigma_{\forall a,k}b_{k\rightarrow a}^{CS}(t)\neq 0$ \textbf{do}\\
4.\quad\quad\textbf{for} each unassociated FUE $k$ \textbf{do}\\
 5.\quad\quad\quad \textbf{while} $D^{\varPhi_k}_{CS}\leq \epsilon$ \textbf{do}\\
               6.\quad\quad \quad\quad\textbf{find} $a=\text{arg max}_{a\in \succ_{k, CS}} \varphi_{CS}^k(a)$.\\
              7.\quad\quad \quad\quad $b_{k\rightarrow a}^{CS}=1$.\\
              8.\quad\quad \quad\quad \:\textbf{for} each node $a$ \textbf{do}\\
              9.\quad\quad\quad \quad \: $\mathcal{L}^{req}_a \leftarrow \{k:b_{k\rightarrow a}^{CS}=1, k\in \mathcal{K}\}$. \\
10.\quad\quad\quad\quad\: construct $\succ_{a,CS}$  using $\varphi_{CS}^a(k)$.\\
              11.\quad\quad\quad\quad\:\textbf{repeat}\\
              12.\quad\quad\quad\quad\quad\:\textbf{if} $|\mathcal{A}_k \cup\{a\}|\leq F_{\text{max}}$\\
13.\quad\quad\quad\quad\quad\quad\: \text{accept}  $k=\text{arg max}_{k\in \succ_{a, CS}} \varphi_{CS}^a(k)$.\\
14.\quad\quad\quad\quad\quad\quad\:  $\mathcal{K}_a := \mathcal{K}_a \cup \{k\}$.\\
15.\quad\quad\quad\quad\quad\quad\:  $\mathcal{A}_k := \mathcal{A}_k \cup \{a\}$.\\
16.\quad\quad\quad\quad\quad\:\textbf{end if}\\
17.\quad\quad\quad\quad\:\textbf{until} $|\mathcal{K}_a|=\hat{N}_a$\\
              18.\quad\quad \quad\quad $\mathcal{L}^{rej}_a:=\mathcal{L}^{rej}_a\setminus  \mathcal{K}_a \, .$\\
              19.\quad\quad \quad\quad remove node $a \in \succ_{k,CS},\forall k \in \mathcal{L}_a^{rej}.$\\
              20.\quad\quad\quad \textbf{end while} \\
              21.\quad \textbf{end while} \\
              22.\textbf{Output:} $\mu_{CS}$.\\
 \label{CTNSA algorithm}\end{algorithm} 
We propose Algorithm \ref{CTNSA algorithm} to perform cooperative node selection in an advanced manner. In this algorithm, after initialization, user $k$ construct its preference list $\mathcal{P}_{k,CS}$ based on the utility function  $\varphi_{CS}^k(a)$ in \eqref{ukcs}. Afterward, user $k$ sends an attachment request $b_{k\rightarrow a}^{CS}$ to the most preferred transmission node $a$. This node adds user $k$ to its request list. Next, node $a$ constructs its preference list based on utility function $\varphi_{CS}^a(k)$ which we have been proposed at \eqref{uacs} before, and checks out the limitations on CS size and the number of connected users.  If these limitations are satisfied, accepts request of user $k$. In the same way, user $k$ sends request to the next preferred node as the other cooperative node until the quality criteria proposed at \eqref{phicri} be satisfied or the CS size be overflowed.
\subsubsection{Matching with externalities}
As the throughput of each
FUE is strongly affected by the dynamic formation of
other FUE-FBSs links due to the dependence of the utility functions
on externalities, the proposed game can be classified as a many-to-many matching game with externalities.
Anomalous to many other papers which work in small cell domain and deals with conventional matching games, we assume that the individual
players’ utilities practically are affected by the other player’s preferences.
In fact, due to externalities, a player may prefer to change its preference order in response to the formation of other UE-SBS links. Therefore, we employ a new stability concept, based on the idea of swap-matching\cite{externalities}. For the given matching $\mu_{CS}$, a swap-matching for any possible pair of FUEs $k$ , $m\in\mathcal{K}$ and FBSs $a$, $i\in\mathcal{A} $ where $(a,k)$, $(i,m)\in \mu_{CS}$, $a\in\mathcal{A}_k $ and $i\in\mathcal{A}_m $   is defined as $\mu^{k}_{{CS}_{i,a}}=\{\mu_{CS}\backslash (a,k)\}\cup (m,k)$. The given matching is stable
if there exist no swap-matchings $\mu^{k}_{{CS}_{i,a}}$ such that $\varphi_{CS}(\mu^{k}_{{CS}_{i,a}})>\varphi_{CS}(\mu_{CS})$ or in the other words, $\mu$ is stable if there is not any transmission node which FUE $k$ prefers to replaced in its accepted set ($\mathcal{A}_k$) and there is not any FUE which node $a$ prefers to serve in its accepted set ($\mathcal{K}_a$). In order to find a stable matching ($\mu^{*}_{CS}$), we propose Algorithm \ref{swap algorithm} which can update matching $\mu_{CS}$ based on the externalities. Line (5) of the algorithm indicates that  FUE $k$ may prefer node $i$ based on the updated utility function. the Algorithm \ref{swap algorithm} monitors any preferred requests based on the given network and matching.
\begin{algorithm}
 \caption{Swap-matching Algorithm.} 
\textbf{\textit{Step 1: Perform initial matching:}}\\
1.\quad\: Import $\mu_{CS}$ through results in Algorithm  \ref{CTNSA algorithm}\\
\textbf{\textit{Step 2: Swap-matching Evaluation:}}\\
2.\quad\textbf{repeat}\\
3.\quad\quad  the utility $\varphi_{CS}$ is updated based on the current $\mu_{CS}$.   \\
4.\quad\quad construct $\succ_{a,CS}$ and $\succ_{k,CS}$ based on the new $\varphi_{CS}$.\\
5.\quad\quad \textbf{if} $(i,\mu^{k}_{{CS}_{i,a}}) \succ_{k,CS} (a,\mu_{CS}) $.\\
6.\quad\quad\quad\: $b_{k\rightarrow i}^{CS}=1$.\\
7.\quad\quad\quad\: node $i$ computes $\varphi_{CS_{i,k}}(\mu^{k}_{{CS}_{i,a}})$.\\
8.\quad\quad\quad\: \textbf{if}  $(k,\mu^{k}_{{CS}_{i,a}}) \succ_{i,CS} (k,\mu_{CS}) $.\\
9.\quad\quad\quad\quad\:  $\mathcal{K}_i := \mathcal{K}_i \cup \{k\}$.\\
10.\quad\quad\quad\quad\:  $\mathcal{A}_k := \mathcal{A}_k \cup \{i\}$.\\
11.\quad\quad\quad\quad\:  $\mu_{CS} \leftarrow \mu^{k}_{{CS}_{i,a}}$\\
12.\quad\quad\quad\:\textbf{end if}\\
13.\quad\quad \textbf{end if}\\
14.\quad\:\textbf{until} $\nexists  (i,\mu^{k}_{{CS}_{i,a}}) \succ_{k,CS} (a,\mu_{CS})  $ and  $(k,\mu^{k}_{{CS}_{i,a}}) \succ_{i,CS} (k,\mu_{CS}) $.\\
              15.\textbf{Output:} $\mu^*_{CS}$.
 \label{swap algorithm}\end{algorithm} 
\subsection{enhanced Subcarrier Allocation (eCA)}
 The subcarrier allocation in PD-NOMA systems is investigated in many researches using the SCA approach or matching theory but none of them has directorship on the resource allocation design such that the interference of NOMA approach can be decreased. To solve this resource allocation problem, we propose a new method based on many-to-one matching structure. 
\begin{subequations} \label{eCA} 
     \begin{align}
&\mathcal{P}_{\text{eCA}} : \hspace{0.25cm} \underset{ \boldsymbol{\nu}}{\text{max}} \quad  \Sigma_{n\in \mathcal{N}}\Sigma_{k\in \mathcal{K}}r_{k,n},  \label{eq:2aa2}  \\
      &\hspace{1.5cm}\text{s.t.}\hspace{0.3cm}\eqref{eq:2e3}, \eqref{eq:2e4}, \eqref{eq:2f}. \nonumber 
     \end{align}
\end{subequations} 
In this problem, $\nu_{k,n}$ indicates the allocation of subcarrier $n$ to user $k$. 
Optimization problem $ \mathcal{P}_{\text{eCA}}$  can be defined by a tuple ($\mathcal{N}, \mathcal{K},  \succ_{\mathcal{N},CA}, \succ_{\mathcal{K},CA}$). Here, $\succ_{\mathcal{N},CA}$ and $\succ_{\mathcal{K},CA}$ denote the sets of the preference relations of FUEs and subcarriers, respectively.

Similar to the CTNSA algorithm, mapping function $\mu_{CA}$ is defined such that:
1) $n\in \mu_{CA}(k) \longleftrightarrow k\in \mu_{CA}(n)$; 
2) $|\mu_{CA}(k)|\leq 1$ and $|\mu_{CA}(n)|\leq q_{\text{max}} $.
As we mentioned at the previous algorithm, we propose some utility functions in order to decrease the impact of the extra-interference of PD-NOMA systems which can not be removed using SIC procedure. The utility function of the FUE $k$ and each subcarrier are described as follows:
\begin{equation} \label{ukcs1} 
     \begin{aligned}
& \varphi_{CA}^k(n)=r_{k,n}=\log_2(1+ \gamma_{k,n}),
     \end{aligned}
\end{equation} 
 \begin{equation} \label{uacs11} 
     \begin{aligned} 
& \varphi_{CA}^n(k)=\varUpsilon_{CA} \frac{r_{k,n}-R_k}{r_{k,n}}-\varTheta_{k,n}.
     \end{aligned}
\end{equation}
To minimize the total interefernce caused by user $i$,  $\{\forall i\in \mathcal{K}   \big{|} \: |\bar{h}^a_{\mathcal{F}_{i,\: n}}|\geq|\bar{h}^a_{\mathcal{F}_{k,\: n}}|   \}$, we employ a new parameter as  $\varTheta_{k,n}=\Sigma_{a\in \mathcal{A}}( \varTheta^n_{\mathcal{MF}_{k,a}}+ \varTheta^n_{\mathcal{F}_{k,a}})$ where 
 $\varUpsilon_{CA}$ is a weighting parameter. 

We use Algorithm \ref{eCA algorithm} to assign subcarriers in an advanced manner. For simplicity, we define the set of users with the same subcarrier as  $\mathcal{C}_n$ where each user constructs its preference list based on achievable data rate and sends an attachment request to all of the nodes in the CS which are determined in CTNSA algorithm. The called subcarrier accepts or rejects this proposal based on its preference list and utility function. If the subcarrier satisfies and the maximum number of users with the same subcarrier does not overflow, it can be assigned to user $k$. After the eCA algorithm,  we employ a swap-matching algorithm similar to Algorithm \ref{swap algorithm}.

 \begin{algorithm}
 \caption{Matching eCA Algorithm.} 
 \textbf{\textit{Step 1: Initialization:}} \\
1.\quad Preset $ \mathcal{L}^{req}_n=\o{}$, $ \mathcal{L}^{rej}_n=\o{}$, $q_{\text{max}}$,  $\forall k, n$. \\
 \textbf{\textit{Step 2: Utility Computation:}}\\
       2.\quad  construct $\mathcal{P}_{k,CA}$  using $\varphi_{CA}^k(n)$ $\forall k$.\\
 \textbf{\textit{Step 3: Find stable matching:}}\\
               3.\quad \textbf{while} $\Sigma_{\forall n,k}b_{k\rightarrow n}^{CA}(t)\neq 0$ \textbf{do}\\
4.\quad\quad\textbf{for} each FUE $k\in \mathcal{K}_a$  \textbf{do}\\
               5.\quad\quad \quad\quad\textbf{find} $n=\text{arg max}_{n\in \succ_{k, CA}} \varphi_{CA}^k(n)$.\\
              6.\quad\quad \quad\quad $b_{k\rightarrow n}^{CA}=1$.\\
              7.\quad\quad \quad\quad \:\textbf{for} each subcarrier $n$ \textbf{do}\\
              8.\quad\quad\quad \quad \: $\mathcal{L}^{req}_n \leftarrow \{k:b_{k\rightarrow n}^{CA}=1, k\in \Sigma_{a\in\mathcal{A}_k}\mathcal{K}_a\}$. \\
9.\quad\quad\quad\quad\: construct $\succ_{n,CA}$  using $\varphi_{CA}^n(k)$.\\
              10.\quad\quad\quad\quad\:\textbf{repeat}\\
 11.\quad\quad\quad\quad\quad\quad\: $k=\text{arg max}_{k\in \succ_{n, CA}} \varphi_{CA}^n(k)$.\\
12.\quad\quad\quad\quad\quad\quad\: \text{assign} $n$ to the \text{FUE}  $k$.\\
13.\quad\quad\quad\quad\quad\quad\:  $\mathcal{C}_n := \mathcal{C}_n \cup \{k\}$.\\ 
14.\quad\quad\quad\quad\:\textbf{until} $|\mathcal{C}_n|=q_\text{max}$\\
              15.\quad\quad \quad\quad $\mathcal{L}^{rej}_n:=\mathcal{L}^{rej}_n\setminus  \mathcal{C}_n \, .$\\
              16.\quad\quad \quad\quad remove subchannel $n \in \succ_{k,CA},\forall k \in \mathcal{L}_n^{rej}.$\\
              17.\quad \textbf{end while} \\
              18.\textbf{Output:} $\mu_{CA}$.
 \label{eCA algorithm}\end{algorithm} 
\subsection{Robust Beamforming design}\label{robustbeam}
After subcarrier allocation step, we try to perform beamforming in imperfect CSI conditions using two different method which are described as follows:
\subsubsection{Worst-Case in No-CSI Situation}														Based on the additive error model and \eqref{ellipsoid}, we define uncertainty of channels in the Euclidean ball-shaped uncertainty sets as $ \mathcal{H}_{\mathcal{F}_{k,\: n}}$, $ \mathcal{H}_{\mathcal{FM}_{k,\: n}}$ and $ \mathcal{H}_{\mathcal{MF}_{ n}}$.
In this case, we define $ \zeta_{k,n} $, $ \kappa_{k,n} $ and $ \eta_{k,n} $ as the error bounds on the uncertainty region of the channel coefficients $ \textbf{h}_{\mathcal{F}_{k,\: n}} $, $ \textbf{h}_{\mathcal{FM}_{k,\: n}} $ and $ \textbf{h}_{\mathcal{MF}_{n}} $, respectively.
In the worst-case scenario,  channel coefficients of problem \eqref{eq:2} must be in the considered uncertainty sets, i.e. $  \textbf{h}_{\mathcal{F}_{k,\: n} }\in  \mathcal{H}_{\mathcal{F}_{k,\: n}}  , \:  \textbf{h}_{\mathcal{FM}_{k,\: n} }\in  \mathcal{H}_{\mathcal{F}_{k,\: n}} , \: \textbf{h}_{\mathcal{MF}_{n} }\in  \mathcal{H}_{\mathcal{F}_{ k,n}} $.
Note that distributions of error vectors are unknown and the critical case of them must be considered. In this regard, instead of $ |(\textbf{w}_{k,\: n}\circ\boldsymbol{\rho}_{k,\: n})^H\textbf{h}_{\mathcal{F}_{k,\: n} }|^{2}$ which is a quadratic function, we can write
\begin{equation}
      \begin{aligned}
    & |\textbf{v}^{H}_{k,\: n}\textbf{h}_{\mathcal{F}_{k,\: n} }|^{2}=\textbf{v}^{H}_{k,\: n}(\bar{\textbf{H}}_{\mathcal{F}_{k,\: n} }+\Delta_{\mathcal{F}_{k,\: n}})\textbf{v}_{k,\: n}=\nonumber\\
&\text{trace}[(\bar{\textbf{H}}_{\mathcal{F}_{k,\: n} }+\Delta_{\mathcal{F}_{k,\: n}})\textbf{V}_{k,\: n} ].
       \label{eq:33}  
      \end{aligned}
\end{equation} 
Since  $|\textbf{v}^{H}_{k,\: n}\textbf{h}_{\mathcal{F}_{k,\: n} }|^{2}$ is a nonlinear expression, we can apply the SDR method where $ \textbf{v}^{H}\textbf{A}\textbf{v}=\text{trace}[\textbf{A}\textbf{v}\textbf{v}^{H}] $.
In this solution, we have confidence that  $ \textbf{V}_{k,n}=(\textbf{w}_{k,\: n}\circ\boldsymbol{\rho}_{k,\: n})*(\textbf{w}_{k,\: n}\circ\boldsymbol{\rho}_{k,\: n})^H=\textbf{W}_{k,n}\circ \boldsymbol{\varrho}_{k,\: n} $ and $ \boldsymbol{\varrho}_{k,\: n}=\boldsymbol{\rho}_{k,\: n}*\boldsymbol{\rho}_{k,\: n}^H $.
 The expression $ \Delta_{\mathcal{F}} $ is defined as a norm-bounded matrix, i.e., $ \|\Delta_{\mathcal{F}}\|\leq \varepsilon_{\mathcal{F}_k}$ and $\bar{\textbf{H}}^f_{\mathcal{F}_{k,\: n} }=\bar{\textbf{h}}^f_{\mathcal{F}_{k,\: n}}\bar{\textbf{h}}^{fH}_{\mathcal{F}_{k,\: n}}$. Hence, $ \varepsilon_{\mathcal{F}_k}  $, $\varepsilon_{\mathcal{FM}_k}$ and $\varepsilon_{\mathcal{MF}_k}$ can be found as follows:
     \begin{align} \label{eq:35}
& \|\Delta_{\mathcal{F}_{k,\: n}}\|\leq \varepsilon_{\mathcal{F}_k}=\zeta^{2}_k + 2\zeta^{2}_k\|\bar{\textbf{h}}_{\mathcal{F}_{k,\: n}}\|,\nonumber \\  
& \|\Delta_{\mathcal{FM}_{k,\: n}}\|\leq \varepsilon_{\mathcal{FM}_k}=\kappa^{2}_k + 2\kappa^{2}_k\|\bar{\textbf{h}}_{\mathcal{FM}_{k,\: n}}\|,\\
& \|\Delta_{\mathcal{MF}_{n}}\|\leq \varepsilon_{\mathcal{MF}}=\eta^{2} + 2\eta^{2}_k\|\bar{\textbf{h}}_{\mathcal{MF}_{n}}\|.\nonumber
  \end{align}
 In order to define the critical situations, 
we apply \eqref{eq:33} with minimizing the numerator and maximizing the denominator of the SINR of users. Hence, \eqref{eq:2b} can be written as
      \begin{align}
      &  \text{trace}[(\bar{\textbf{H}}_{\mathcal{F}_{k,\: n} }-\varepsilon_{\mathcal{F}_{k}}\textbf{I}_{ F T_f})(\textbf{W}_{k,n}\circ\boldsymbol{\varrho}_{k,\: n}) ]- \nonumber \\
&(2^{R_k}-1) (I^{W}_{\mathcal{FM}_{k,\: n}}+I^{W}_{\mathcal{F}_{k,\: n}})\geq (2^{R_k}-1)\sigma_{k,n}^2, \label{eq:38}
      \end{align}
where $ I^{W}_{\mathcal{F}_{k,\: n}}=\Sigma_{i\neq k,\|\textbf{h}_{\mathcal{F}_{i,\: n}}\|> \|\textbf{h}_{\mathcal{F}_{k,\: n}} \|}   \text{trace}[(\bar{\textbf{H}}_{\mathcal{F}_{n} }+\varepsilon_{\mathcal{F}_{n}}\textbf{I}_{ FT_f})(\textbf{W}_{i,n}\circ\boldsymbol{\varrho}_{i,\: n}) ] $ and $  I^{W}_{\mathcal{FM}_{k,\: n}}= \text{trace}[(\bar{\textbf{H}}_{\mathcal{FM}_{k,\: n} }+\varepsilon_{\mathcal{FM}_{k}}\textbf{I}_{T_m})\textbf{M}_{n} ] $.
َAs the same way, the critical equivalent of \eqref{eq:2c} can be expressed as 
\begin{align}
   \Sigma_{k \in \mathcal{K}} \text{trace}[(\bar{\textbf{H}}_{\mathcal{MF}_{n} }+\varepsilon_{\mathcal{MF}_{n}}\textbf{I}_{FT_f})(\textbf{W}_{k,n}\circ\boldsymbol{\varrho}_{k,\: n}) ] \leq \epsilon_M.   \label{eq:39}   
\end{align}
The final problem in the worst-case scenario is rewritten as follows:
\begin{subequations} \label{eq:40}
     \begin{align}
     & \underset{\textbf{W} ,\:\boldsymbol{\varrho}}{\text{max}} \quad  \Sigma_{k\in \mathcal{K}}\Sigma_{n\in \mathcal{N}}   r_{k,n},  \\
      &\text{s.t.} \hspace{.25cm} \eqref{sic}, \eqref{eq:38} \text{ and } \eqref{eq:39}\nonumber \\
               &\hspace{1cm} 0 < \Sigma_{k\in \mathcal{K}}\Sigma_{n\in \mathcal{N}} \text{trace}[\hat{\textbf{W}}^f_{k,\: n}\circ\hat{\boldsymbol{\varrho}}^f_{k,\: n}]\leq P_\text{max}, \label{eq:sdrd}\\
&\hspace{1cm} \textbf{W}_{k,\:n}\succeq 0, \;\text{rank}(\textbf{W}_{k,\:n} ) = 1.  \label{eq:sdrh}          
     \end{align}
\end{subequations}
Due to the non-convex rate function, the optimization problem \eqref{eq:40} is non-convex.
To tackle this issue, the SCA approach with difference of two concave functions (D.C.) approximation method
is used. At first, the 
rate function in the objective is written as
\begin{align}
& r_{k,n}=f_{k,n} - g_{k,n}, \label{dfg}
\end{align}
where 
\begin{align}
& f_{k,n}=\log_2\big(\text{trace}[(\bar{\textbf{H}}_{\mathcal{F}_{k,\: n} }-\varepsilon_{\mathcal{F}_{k,n}}\textbf{I}_{ F T_f})(\textbf{W}_{k,n}\circ\boldsymbol{\varrho}_{k,\: n}) ]+\nonumber \\ & I^{W}_{\mathcal{FM}_{k,\: n}}+I^{W}_{\mathcal{F}_{k,\: n}} +\sigma_{k,n}^2\big),  \label{f}\\
& g_{k,n}=\log_2(  I^{W}_{\mathcal{FM}_{k,\: n}}+I^{W}_{\mathcal{F}_{k,\: n}} +\sigma_{k,n}^2). \label{g}
\end{align}
By applying the D.C. approximation, $ g_{k,n} $ is approximated as follows
\begin{align}
& g_{k,n}(\textbf{W}_{k,n})\approx  g_{k,n}(\textbf{W}_{k,n}^{[t-1]}) + \nonumber \\ &\langle \nabla g_{k,n}(\textbf{W}_{k,n}^{[t-1]}),(\textbf{W}_{k,n}^{[t]}-\textbf{W}_{k,n}^{[t-1]})\rangle, \label{appg}
\end{align}
where 
\begin{align}\label{g gradian Downlink BS}
&\nabla g_{k,n}(\textbf{W}_{k,n})  =\\& \nonumber
\left\{
  \begin{array}{ll}
    0, & \hbox{$\forall i = k$}, \\
    \frac{(\bar{\textbf{H}}_{\mathcal{F}_{n} }+\varepsilon_{\mathcal{F}_{n}}\textbf{I}_{ FT_f})^T\circ\boldsymbol{\varrho}_{i,\: n}}  {\ln(2) \big( I^{W}_{\mathcal{FM}_{k,\: n}} + I^{W}_{\mathcal{F}_{k,\: n}} +\sigma_{k,n}^2 \big)}, & \hbox{$\forall i \in \mathcal{K}, \|\bar{\textbf{h}}_{\mathcal{F}_{i,\: n}}\|> \|\bar{\textbf{h}}_{\mathcal{F}_{k,\: n}} \|$}.
  \end{array}
\right.
\end{align}
To approximate \eqref{sic} as the SIC constraint, we  employ a combinatorial method. 
At the first stage, we apply the D.C. approximation on the SINR functions which are at two sides of the SIC inequality as follows:
\begin{equation} \label{ss1} 
     \begin{aligned} 
&  
\Gamma_{(\textbf{w}_{k,\:n},\ \boldsymbol{\rho}_{k,\: n},\ \textbf{h}_{\mathcal{F}_{j,\: n}},\: \textbf{h}_{\mathcal{FM}_{j,\:n}})} = f^{'}_{k,n} - g^{'}_{k,n},
     \end{aligned}
\end{equation} 
 \begin{equation} \label{ss2} 
     \begin{aligned} 
& 
\Gamma_{(\textbf{w}_{k,\:n},\ \boldsymbol{\rho}_{k,\: n},\ \textbf{h}_{\mathcal{F}_{k,\: n}},\: \textbf{h}_{\mathcal{FM}_{k,\:n}})} = f^{''}_{k,n} - g^{''}_{k,n},
     \end{aligned}
\end{equation}
where $ f^{'}_{k,n} $ and $ g^{'}_{k,n} $ are the numerator and denominator
 of $ \Gamma_{(\textbf{w}_{k,\:n},\ \boldsymbol{\rho}_{k,\: n},\ \textbf{h}_{\mathcal{F}_{j,\: n}},\: \textbf{h}_{\mathcal{FM}_{j,\:n}})} $, respectively. Gradient of $ g^{'}_{k,n} $ are described as follow
\begin{align}\label{1g gradian Downlink BS}
&\nabla^{'}_{g^{'}_{k,n}} (\textbf{W}_{k,n}, \textbf{h}_{\mathcal{F}_{j,\: n}})  =\alpha_{i,j}(\textbf{h}_{\mathcal{F}_{j,\: n}}\textbf{h}^{T}_{\mathcal{F}_{j,\: n}})\circ\boldsymbol{\varrho}_{i,\: n},    
\end{align}
where $ \alpha_{i,j} $ equals one $ \forall i \in \mathcal{K}, \|\bar{\textbf{h}}_{\mathcal{F}_{i,\: n}}\|> \|\bar{\textbf{h}}_{\mathcal{F}_{j,\: n}} \| $ and equals zero for others.
As \eqref{ss1} and \eqref{ss2} are similar functions, \eqref{ss2} is approximated by the D.C. solution where $ \nabla^{''}_{g^{''}_{k,n}} (\textbf{W}_{k,n}, \textbf{h}_{\mathcal{F}_{k,\: n}}) $ can be calculated like \eqref{1g gradian Downlink BS}. 
 Next, we apply \eqref{additive} and Euclidean ball-shaped uncertainty set on the approximated SIC constraint.  Accordingly, inequality \eqref{sic} changes as follows  
\begin{align}\label{sic_worst}
& \text{trace}\big[(\bar{\textbf{H}}_{\mathcal{F}_{j,\: n} }-\bar{\textbf{H}}_{\mathcal{F}_{k,\: n} }-(\varepsilon_{\mathcal{F}_{j,\: n}}+\varepsilon_{\mathcal{F}_{k,\: n}})\textbf{I}_{F T_f})(\textbf{W}_{k,n}\circ\boldsymbol{\varrho}_{k,\: n}) \big]\nonumber \\
& + \text{trace}\big[\big(\bar{\textbf{H}}_{\mathcal{F}_{j,\: n}} +\varepsilon_{\mathcal{F}_{j,\: n}}\textbf{I}_{F T_f}\big) \textbf{T}_{(\textbf{W}^{[t]}_{k,n}, \textbf{h}_{\mathcal{F}_{j,\: n}},\alpha_{i,j})}
\big]\nonumber \\
& - \text{trace}\big[\big(\bar{\textbf{H}}_{\mathcal{F}_{k,\: n}} -\varepsilon_{\mathcal{F}_{k,\: n}}\textbf{I}_{F T_f}\big) \textbf{T}_{(\textbf{W}^{[t]}_{k,n}, \textbf{h}_{\mathcal{F}_{k,\: n}},\alpha_{i,k})}\big]\nonumber \\
& \text{trace}[\big(\bar{\textbf{H}}_{\mathcal{FM}_{k,\: n} }-\bar{\textbf{H}}_{\mathcal{FM}_{j,\: n} }-(\varepsilon_{\mathcal{FM}_{k,\: n}}+\varepsilon_{\mathcal{FM}_{j,\: n}})\textbf{I}_{T_m}\big)\textbf{M}_{n}]\nonumber \\ &+\sigma^2_{k,n}-\sigma^2_{j,n}\geq 0, \forall j,k \in \mathcal{K}, n\in \mathcal{N}.
\end{align}
 where 
 \begin{align}
 & \textbf{T}_{(\textbf{W}^{[t]}_{k,n}, \textbf{h}_{\mathcal{F}_{j,\: n}},\alpha_{i,j})}=\big((-\Sigma_{i\in \mathcal{K},\|\textbf{h}_{\mathcal{F}_{i,\: n}}\|> \|\textbf{h}_{\mathcal{F}_{j,\: n}} \|} \textbf{W}^{[t-1]}_{i,n}+\nonumber \\ &\alpha_{i,j}\textbf{W}^{[t-1]}_{k,n}-\alpha_{i,j}\textbf{W}^{[t]}_{k,n})\circ\boldsymbol{\varrho}_{i,\: n}\big)
 \end{align}
Since rank(.) is a non-convex constraint, it can be guaranteed with Gaussian randomization method.
By applying the D.C. approximation, the optimization problem \eqref{eq:40} is approximated by a convex function which can be solved by CVX toolbox.       
 \subsubsection{Stochastic Imperfect CSI Case}
 In this case, we consider that distribution and covariance of the error vectors are clear and because of the independence between antennas in different FBSs, error vectors from all of the FBSs and MBS to the $ k^{th} $ user at the $ n^{th} $ subcarrier can be expressed as \eqref{bbb}.
Further, we can use the definition of cumulative distribution function (CDF) of exponential random variable (i.e.,$ f(x)=\lambda e^{-\lambda x}  $ for $ 0\leq x $ and $ F_{(x;\lambda)}= 1-e^{-\lambda x} $) for \eqref{eq:2c}. In the paper, we assume $ \lambda=1 $ and $ x=\epsilon_M\times\text{trace}[\textbf{C}_{\textbf{h},\mathcal{MF}_{ n}}\,(\textbf{W}_{k,n}\circ\boldsymbol{\varrho}_{k,\: n})]\,  $. Hence, constraint \eqref{eq:2c} can be rewritten as
\begin{align}\label{eq:cdf}
 \Sigma_{k\in \mathcal{K}} \text{trace}[\textbf{C}_{\textbf{h},\mathcal{MF}_{ n}}\,(\textbf{W}_{k,n}\circ\boldsymbol{\varrho}_{k,n})]\leq   \frac{\epsilon_M}{\ln \frac{1}{\alpha}}.
\end{align}
 Let $ \bar{\textbf{h}}_{k,\:n}$, $\textbf{v}_{k,\:n}$ and $\textbf{C}^{1/2}_{e,\:k,\:n}$ be defined as follows:
 \begin{equation}
 \bar{\textbf{h}}_{k,\:n}= \left(
 \begin{matrix}
 \bar{\textbf{h}}_{\mathcal{F}_{k,\:n}}\\
 \bar{\textbf{h}}_{\mathcal{FM}_{k,\:n}}
 \end{matrix}
 \right)\:,\: \textbf{v}_{k,\:n}= \left(
 \begin{matrix}
 \textbf{v}_{\mathcal{F}_{k,\:n}}\\
 \textbf{v}_{\mathcal{FM}_{k,\:n}}
 \end{matrix}
  \right)  , 
 \label{eq:tarkib matrix}
 \end{equation}

 \begin{equation}
 \textbf{C}^{1/2}_{e,\:k,\:n}= \left(
   \begin{matrix}
 \textbf{C}^{1/2}_{e,\mathcal{F}_{k,\:n}}&\textbf{0}\\
 \textbf{0}&\textbf{C}^{1/2}_{e,\mathcal{FM}_{k,\:n}}
   \end{matrix}
   \right),
 \end{equation}
where $ \textbf{v}_{k,\:n}\sim \mathcal{N}(\textbf{0},\textbf{I}_{ FT_f+T_m})$. 
If we apply \eqref{eq:cdf} and the SDR method,  \eqref{eq:2} can be expressed as follows:
   \begin{subequations} \label{eq:sdr}
        \begin{align}
         &\underset{\textbf{W} ,\:\boldsymbol{\rho}}{\text{max}} \quad \Sigma_{k\in \mathcal{K}}\Sigma_{n\in \mathcal{N}}  r_{k,n}, \label{eq:sdra} \\
         &\text{s.t.}\hspace{0.25cm} \eqref{sic}, \eqref{eq:sdrd}, \eqref{eq:sdrh},  \eqref{eq:cdf}\nonumber \\
          & \hspace{1cm} \Pr {\bigg \{ } \textbf{v}^{H}_{k,\:n}\textbf{A}_{(\textbf{W}_{k,n}, \boldsymbol{\varrho}_{k,\: n})}\textbf{v}_{k,\:n}+2\text{Re}\{\textbf{v}^{H}_{k,\:n} \textbf{b}_{(\textbf{W}_{k,n},\boldsymbol{\varrho}_{k,\: n})}\}\nonumber \\& \geq c_{(\textbf{W}_{k,n},\boldsymbol{\varrho}_{k,\: n})} {\bigg \}} 
               \geq 1 -\beta, \label{eq:sdrb}
        \end{align} 
   \end{subequations} 
 where $ \textbf{A}_{(\textbf{W}_{k,n}, \boldsymbol{\varrho}_{k,\: n})} $, $ \textbf{b}_{(\textbf{W}_{k,n}, \boldsymbol{\varrho}_{k,\: n})} $ and $ c_{(\textbf{W}_{k,n}, \boldsymbol{\varrho}_{k,\: n})} $ are defined as follows:
     \begin{equation}
         \begin{aligned} 
 & \textbf{A}_{(\textbf{W}_{k,n}, \boldsymbol{\varrho}_{k,\: n})}=\textbf{C}^{1/2}_{e,\:k,\:n} \textbf{W}_{T_{k,\:n}}\textbf{C}^{1/2}_{e,\:k,\:n} , \\
  & \textbf{b}_{(\textbf{W}_{k,n}, \boldsymbol{\varrho}_{k,\: n})}=\textbf{C}^{1/2}_{e,\:k,\: n} \textbf{W}_{T_{k,\:n}}\bar{\textbf{h}}_{k,\: n} ,\\
 &  c_{(\textbf{W}_{k,n}, \boldsymbol{\varrho}_{k,\: n})}=-\bar{\textbf{h}}_{k,\: n}^{H} \textbf{W}_{T_{k,\:n}} \bar{\textbf{h}}_{k,\: n}+\sigma_{k,n}^2 , \label{eq:23} \\ 
      \end{aligned}
 \end{equation}
 and 
 \begin{equation}
 \textbf{W}_{T_{k,\:n}}= \left(
 \begin{matrix}
\textbf{D}_{\mathcal{F}_{k, n}}&\textbf{0}\\
\textbf{0}&-\textbf{m}_{n}\textbf{m}_{n}^H
 \end{matrix}
 \right), \label{d25}
 \end{equation}
where $ \textbf{D}_{\mathcal{F}_{k, n}} =\frac{1}{2^{R_k}-1}\textbf{W}_{k,n}\circ\boldsymbol{\varrho}_{k,\: n}-\Sigma_{i\in \mathcal{K},\|\textbf{h}_{\mathcal{F}_{i,\: n}}\|> \|\textbf{h}_{\mathcal{F}_{k,\: n}} \|}\textbf{W}_{i,n}\circ\boldsymbol{\varrho}_{i,\: n} $. Since \eqref{eq:sdrb} is a probabilistic inequality, we can use the Bernstein-Type inequality for quadratic forms of Gaussian variables. Therefore, \eqref{eq:sdrb} can replaced as follows:
  \begin{subequations}
           \begin{align}
           &
 \text{trace}(\textbf{A}_{(\textbf{W}_{k,n},\boldsymbol{\varrho}_{k,\: n})})-\sqrt{2\xi}x -\xi y \geq c_{(\textbf{W}_{k,n},\boldsymbol{\varrho}_{k,\: n})}, \label{eq:24}\\
 & \sqrt{\|\textbf{A}_{(\textbf{W}_{k,n}, \boldsymbol{\varrho}_{k,\: n})}\|^{2}_{F}+2\|\textbf{b}_{(\textbf{W}_{k,n}, \boldsymbol{\varrho}_{k,\: n})}\|^{2}}\leq x_{k,n},\label{eq:241}\\
 & y_{k,n}\textbf{I}_{FT_f+T_m}+\textbf{A}_{(\textbf{W}_{k,n}, \boldsymbol{\varrho}_{k,\: n})}\succeq 0, \ y_{k,n}\geq 0, \label{eq:242}
      \end{align}
 \end{subequations}
where $ y=\max \{\lambda_{\text{max}}(-\textbf{A}),0\} $, i.e., $y$ is the maximum eigenvalue of the matrix (-\textbf{A}) and $y$ and $ x $ are slack variables and $ \xi=-\ln \beta $.
To minimize the transmit power, $ y $ must be the principal eigenvalue of  $\textbf{C}_{e,FM}^{1/2}\textbf{m}_n\textbf{m}_n^H\textbf{C}_{e,FM}^{1/2}$, i.e., $ y=\| \textbf{C}_{e,FM}^{1/2}\textbf{m}_n \|$\cite{wang2013het}. Therefore, \eqref{eq:24}-\eqref{eq:242}
can be rewritten as  follows:
\begin{subequations}
\begin{align}
 & \text{trace}[(\textbf{C}_{\textbf{e},\mathcal{F}_{k,\: n}}+\bar{\textbf{h}}_{\mathcal{F}_{k,\:n}}\bar{\textbf{h}}_{\mathcal{F}_{k,\:n}}^{H})\textbf{D}_{\mathcal{F}_{k, n}}]- \sqrt{2\xi}x_{k,n} \nonumber \\ 
& \geq \sigma_{k,n}^2+ \textbf{m}_{n}^H((1+\xi)\textbf{C}_{\textbf{e},\mathcal{FM}_{k,\: n}}+\bar{\textbf{h}}_{\mathcal{FM}_{k,\:n}}\bar{\textbf{h}}_{\mathcal{FM}_{k,\:n}}^H)\textbf{m}_{n}, \label{sinr_1} \\ 
         &
         \left\|
         \begin{matrix}
         vec(\textbf{C}_{\textbf{e},\mathcal{F}_{k,\: n}}^{1/2}\textbf{D}_{\mathcal{F}_{k, n}}\textbf{C}_{\textbf{e},\mathcal{F}_{k,\: n}}^{1/2}) \\
         \sqrt{2}vec(\textbf{C}_{\textbf{e},\mathcal{F}_{k,\: n}}^{1/2}\textbf{D}_{\mathcal{F}_{k, n}}\bar{\textbf{h}}_{\mathcal{F}_{k,\: n}})\\
        \sqrt{ \varsigma_{\mathcal{FM}_{k,n}}}
         \end{matrix}
         \right\|\leq x_{k,n},\label{sinr_2} 
\end{align}
\end{subequations}
where 
 \begin{align}
&   \varsigma_{FM_{k,n}} \equiv  \|\textbf{C}_{\textbf{e},\mathcal{FM}_{k,\: n}}^{1/2}\textbf{m}_{n}\textbf{m}_{n}^H\textbf{C}_{\textbf{e},\mathcal{FM}_{k,\: n}}^{1/2}\|^2_{F}+\nonumber \\ & 2\|\textbf{C}_{\textbf{e},\mathcal{FM}_{k,\: n}}^{1/2}\textbf{m}_{n}\textbf{m}_{n}^H\bar{\textbf{h}}_{\mathcal{FM}_{k,\: n}}\|^2.\label{eq:26-1}
 \end{align}
 By applying \eqref{sinr_1} and \eqref{sinr_2}, problem \eqref{eq:sdr} can be rewritten as follows:
 \begin{subequations}\label{eq:finberns1}
      \begin{align}
       &\underset{\textbf{W} ,\:\boldsymbol{\rho},\textbf{X},y}{\text{max}} \quad  \Sigma_{k\in \mathcal{K}}\Sigma_{n\in \mathcal{N}}  r_{k,n},  \\
      & \text{s.t.}\hspace{0.25cm}  \eqref{sic}, \eqref{eq:sdrd}, \eqref{eq:sdrh},  \eqref{eq:cdf},\: \eqref{sinr_1}\: \text{and}\: \eqref{sinr_2}. \nonumber
  \end{align}
 \end{subequations}
The non-convex rate function can be approximated by \eqref{dfg} as follows 
\begin{subequations}
\begin{align}
& f_{k,n}=\log_2\big(\textbf{h}_{\mathcal{F}_{k,\: n}}^T(\textbf{W}_{k,n}\circ\boldsymbol{\varrho}_{k,\: n})\textbf{h}_{\mathcal{F}_{k,\: n}}+\nonumber \\ &  I_{\mathcal{F}_{k,\: n}} +   I_{\mathcal{FM}_{k,\: n}}+\sigma_{k,n}^2\big),  \label{f1}\\
& g_{k,n}=\log_2\big( I_{\mathcal{F}_{k,\: n}} +   I_{\mathcal{FM}_{k,\: n}}+\sigma_{k,n}^2\big),  \label{g1}
\end{align}
\end{subequations}
where $  I_{\mathcal{FM}_{k,\: n}}=  \textbf{h}_{\mathcal{FM}_{k,\: n}}^T\textbf{M}_{k,n}\textbf{h}_{\mathcal{FM}_{k,\: n}} $
and $ I_{\mathcal{F}_{k,\: n}}=\Sigma_{i\neq k,\|\textbf{h}_{\mathcal{F}_{i,\: n}}\|\geq \|\textbf{h}_{\mathcal{F}_{k,\: n}} \|} \textbf{h}_{\mathcal{F}_{k,\: n}}^T(\textbf{W}_{i,n}\circ\boldsymbol{\varrho}_{i,\: n})\textbf{h}_{\mathcal{F}_{k,\: n}} $ are replaced based on \eqref{additive} and \eqref{bbb}. 
Approximation of $ g_{k,n} $ is calculated as \eqref{appg} where 
\begin{align}\label{g gradian Downlink BS1}
&\nabla g_{k,n}(\textbf{W}_{k,n})  =\\& \nonumber
\left\{
  \begin{array}{ll}
    0, & \hbox{$\forall i = k$}, \\
    \frac{(\textbf{h}_{\mathcal{F}_{k,\: n}}\textbf{h}_{\mathcal{F}_{k,\: n}}^T)\circ \boldsymbol{\varrho}_{i,\: n} }  {\ln(2) \big(I_{\mathcal{FM}_{k,\: n}} + I_{\mathcal{F}_{k,\: n}} +\sigma_{k,n}^2 \big)}, & \hbox{$\forall i \neq k, \|\bar{\textbf{h}}_{\mathcal{F}_{i,\: n}}\|\geq \|\bar{\textbf{h}}_{\mathcal{F}_{k,\: n}} \|$}.
  \end{array}
\right.
\end{align}

For the non-convex SIC constraint, we apply ‌Bernstein solution on \eqref{ss1}, \eqref{ss1} and \eqref{1g gradian Downlink BS} as follows 
\begin{align}\label{bern2}
 &\Pr {\bigg \{ } \textbf{v}^{H}_{j,\:n}\textbf{A}^{'}_{( \textbf{W}^{'}_{T_{k,\:n}}, \textbf{C}_{e,\:j,\:n})} \textbf{v}_{j,\:n}+2\text{Re}\{\textbf{v}^{H}_{j,\:n}\textbf{b}^{'}_{( \textbf{W}^{'}_{T_{k,\:n}}, \textbf{C}_{e,\:j,\:n})}\} \nonumber \\
 &
 -\textbf{v}^{H}_{k,\:n}\textbf{A}^{''}_{( \textbf{W}^{''}_{T_{k,\:n}}, \textbf{C}_{e,\:k,\:n})}\textbf{v}_{k,\:n}-2\text{Re}\{\textbf{v}^{H}_{k,\:n}\textbf{b}^{''}_{( \textbf{W}^{''}_{T_{k,\:n}}, \textbf{C}_{e,\:k,\:n})}\} \nonumber \\
 &
 \geq c^{'}_{( \textbf{W}^{'}_{T_{k,\:n}}, \textbf{C}_{e,\:j,\:n})}-c^{''}_{( \textbf{W}^{''}_{T_{k,\:n}}, \textbf{C}_{e,\:k,\:n})} {\bigg \}} 
               \geq 1 -\beta,\forall j, k, n,
\end{align}
where $ \textbf{A}^{'} $, $ \textbf{A}^{''} $, $ \textbf{b}^{'} $, $ \textbf{b}^{''} $ , $ c^{'} $ and $ c^{''} $ are acquired similar to \eqref{eq:23}. $ \textbf{W}^{'}_{T_{k,\:n}}$ is defined with the same structure of \eqref{d25} where instead of $\textbf{D}_{\mathcal{F}_{k, n}}$ we define $\textbf{D}^{'}_{\mathcal{F}_{k, n}}=\textbf{W}^{[t]}_{k,n}\circ\boldsymbol{\varrho}_{k,\: n}- \textbf{T}_{(\textbf{W}^{[t]}_{k,n}, \textbf{h}_{\mathcal{F}_{j,\: n}},\alpha_{i,j})}$. Hence, $ \textbf{W}^{''}_{T_{k,\:n}} $ is alike $ \textbf{W}^{'}_{T_{k,\:n}} $ based on  $\textbf{D}^{''}_{\mathcal{F}_{k, n}}$ and $ \textbf{T}_{(\textbf{W}^{[t]}_{k,n}, \textbf{h}_{\mathcal{F}_{k,\: n}},\alpha_{i,k})} $. After applying the Bernstein inequality with a new variable as $ x^{'} $, we have
  \begin{subequations}
           \begin{align}
& \text{trace}[(\textbf{C}_{\textbf{e},\mathcal{F}_{j,\: n}}+\bar{\textbf{h}}_{\mathcal{F}_{j,\:n}}\bar{\textbf{h}}_{\mathcal{F}_{j,\:n}}^{H})\textbf{D}^{'}_{\mathcal{F}_{k, n}}]
- \nonumber \\ &\text{trace}[(\textbf{C}_{\textbf{e},\mathcal{F}_{k,\: n}}+\bar{\textbf{h}}_{\mathcal{F}_{k,\:n}}\bar{\textbf{h}}_{\mathcal{F}_{k,\:n}}^{H})\textbf{D}^{''}_{\mathcal{F}_{k, n}}]\nonumber \\ 
 & -\sqrt{2\xi}x^{'}_{k,n}  \geq \sigma_{i,n}^2-\sigma_{k,n}^2+ \nonumber \\ & \textbf{m}_{n}^H((1+\xi)\textbf{C}_{\textbf{e},\mathcal{FM}_{j,\: n}}+\bar{\textbf{h}}_{\mathcal{FM}_{j,\:n}}\bar{\textbf{h}}_{\mathcal{FM}_{j,\:n}}^H)\textbf{m}_{n}\nonumber \\
 & -\textbf{m}_{n}^H(\textbf{C}_{\textbf{e},\mathcal{FM}_{k,\: n}}+\bar{\textbf{h}}_{\mathcal{FM}_{k,\:n}}\bar{\textbf{h}}_{\mathcal{FM}_{k,\:n}}^H)\textbf{m}_{n} ,\forall j,k\in \mathcal{K}, n\in \mathcal{N},  \label{berbsic1}\\
 & \sqrt{\vartheta}\leq x^{'}_{k,n},\forall j,k\in \mathcal{K}, n\in \mathcal{N},\label{berbsic2}
      \end{align}
 \end{subequations}
 as the set of constraints that must be employed instead of the probabilistic SIC constraint where
$ \vartheta= \|\textbf{A}^{'}\|^{2}_{F}+2\|\textbf{b}^{'}\|^{2}+\|\textbf{A}^{''}\|^{2}_{F}+2\|\textbf{b}^{''}\|^{2} $.
If we employ replacement constraints, 
 problem \eqref{eq:finberns1} can be efficiently solved with the CVX software package. 
 \subsection{The Joint Resource Allocation Algorithm}\label{Joint}
In a centralized network, BSs can be connected to a central unit via backhaul links. We aim to employ a structure based on Algorithm \ref{GPsubcarrier} to allocate resources in this central unit. After performing measurements about channel quality and traffic and determining QoS requirements, this information  acts as the  input of Algorithm \ref{GPsubcarrier} in the central unit which is established as a central controller. We express our framework like a multi-stage attachment procedure  with some criteria and limitations to perform accept/reject decision in each stage. These stages are performed at the central unit considering interests of each user through a utility function criterion.
  The central unit determines the set of accepted users for all transmission nodes through the first stage of  Algorithm \ref{GPsubcarrier}. Based on the available subcarriers in each transmission node, this unit defines some criteria considering overal throughput of the network. Also, if each subcarrier is assigned to  user, the central unit performs suitable beamforming design for this user which is accepted successfully at last of the procedure.
 To solve resource allocation problem \eqref{eq:2} through Algorithm \ref{GPsubcarrier}, after initialization, communications between users and CSs are classified through variable $\boldsymbol{\chi}^*$. Then the subcarrier indicator can be obtained via the eCA algorithm as variable $\boldsymbol{\nu}^*$. The beamforming vector for each FBS is calculated using SCA and Gaussian randomization methods.
 This sequential iterative procedure continues until converged.

 \begin{algorithm}
 \caption{Alternative Sequential Algorithm.} 
 \textbf{\textit{Step 1: Initialization:}} \\
1.\quad Choose a feasible $ \textbf{W} $, $\boldsymbol{\nu} $, $ \epsilon_c $, and $  \epsilon_M  $.  \\
\textbf{\textit{Find Stable ASM Solution:}} \\
2.\quad \textbf{while} $   |\Sigma_{k\in \mathcal{K}}\Sigma_{n\in \mathcal{N}}  r_{k,n}^{(l)}-\Sigma_{k\in \mathcal{K}}\Sigma_{n\in \mathcal{N}}  r_{k,n}^{(l-1)}|\leq \epsilon_c   $ \textbf{do}\\
 \textbf{\textit{Step 2: {Cooperative Set Update and Clustering:}}}\\
       3.\quad  To find $\boldsymbol{\chi}^*(l)$,  solve \eqref{CTNSA} for preset $ \textbf{W} $ and $\boldsymbol{\nu} $.\\
 \textbf{\textit{Step 3: Subcarrier Allocation Update:}}\\
               4.\quad To find $\boldsymbol{\nu}^*(l)$,  solve \eqref{eCA} for preset $ \textbf{W} $ and $\boldsymbol{\chi}^*(l) $.\\
 \textbf{\textit{Step 4: Beamforming Design Update:}}\\
               5.\quad  To find $\textbf{W}^*(l)$,  solve \eqref{eq:40} and \eqref{eq:finberns1} for fix $ \boldsymbol{\chi}^*(l)  $ and $\boldsymbol{\nu}^*(l)$. \\
6.\quad \textbf{end while} \\
 \textbf{Output}  $ \textbf{W}^* $, $\boldsymbol{\chi}^* $ and  $\boldsymbol{\nu}^*$.
 \label{GPsubcarrier}\end{algorithm} 

 \subsection{Joint  transmission node and subcarrier allocation in a SCA algorithm}\label{subcarrier_SCA}
 \begin{subequations} \label{eq:subsca}
  In this subsection, we solve the joint cooperative node selection and subcarrier allocation problem with a SCA  based approach similar to the beamforming problem. Therefore, we rewrite problem \eqref{eq:2} as follows:
      \begin{align}
       & \underset{\textbf{W} , \boldsymbol{\rho}}{\text{max}} \quad  \Sigma_{n\in \mathcal{N}}\Sigma_{k\in \mathcal{K}}r_{k,n},   \\
       &\text{s.t.}\hspace{0.3cm}\eqref{eq:2b}, \eqref{sic}, \eqref{eq:2c}, \eqref{eq:2d}  \nonumber \\
  &\hspace{1cm} \Sigma_{a\in \mathcal{A}}\rho^{a}_{k,n}\leq F_\text{max}, \forall k\in \mathcal{K},  \forall n\in \mathcal{N} \label{limitfmax} \\
 &\hspace{1cm} \Sigma_{n\in \mathcal{N}}\Sigma_{k\in \mathcal{K}}\rho^{a}_{k,n}\leq \hat{N}_a, \forall a\in \mathcal{A}, \label{limitna} \\
               &\hspace{1cm} \Sigma_{k\in \mathcal{K}}\rho^a_{k,n}\leq q_\text{max}, \forall a\in \mathcal{A}, \forall n\in \mathcal{N}, \label{eq:2e3sca} \\
 &\hspace{1cm} \Sigma_{n\in \mathcal{N}}\rho^a_{k,n}\leq 1, \forall a\in \mathcal{A}, \forall k\in \mathcal{K}, \label{eq:2e4sca} \\
              & \hspace{1cm} \rho^a_{k,n}\in
              \{0, 1\},
                              \forall k\in \mathcal{K}, n\in \mathcal{N}, a\in \mathcal{A},                 
               \label{eq:2fsca} 
      \end{align}
 \end{subequations}
 where $ \rho $ has been introduced, before.
 Considering the proposed method for the beamforming problem in subsection \ref{robustbeam}, we employ \eqref{eq:40} and \eqref{eq:finberns1} to manage constraints \eqref{eq:2b}, \eqref{sic}, \eqref{eq:2c}, and \eqref{eq:2d} in both the worst-case and probabilistic approaches, respectively. Moreover, constraints 
 \eqref{limitfmax}-\eqref{eq:2fsca} can be satisfied assuming variable $ \boldsymbol{\varrho}_{k,\: n} $ where each element of the main diameter of matrix $ \boldsymbol{\varrho}_{k,\: n} $ indicates $ \rho^a_{k,n} $.
\subsection{Convergence and stability of the iterative algorithm}\label{secconvergency}

In order to prove convergence of the proposed solution, we introduce Theorem  \ref{theorem1} as follows:
       \begin{theorem} \label{theorem1}
       Matchings in the CTNSA and eCA algorithms are
       stable in each iteration of Algorithm \ref{GPsubcarrier}.
       \end{theorem} 
       \begin{proof}
       As the proposed matching algorithms of this manuscript are based on basic principles
       of the deferred-acceptance algorithm and college admissions
       model with responsive preferences, a property of the stable algorithm in the matching theory is defined as: there is any node that can join to or remove from the stable group which the algorithm calculated before \cite{maingame} and \cite{stable1}.
                      To prove this statement, we consider a stable group $ \mathcal{G}_n $ of $(k,a)$ pairs which is introduced at the end of  iteration $l$. Then, we try to check that FUE $\acute{k}$, $\acute{k}\neq k$, can not join to  $ \mathcal{G}_n $.  Because of the preference relation introduced in algorithms CTNSA and eCA, FUE $k$ is the most preferred one which overcom to  $\acute{k}$ and $\Sigma_{(\acute{k}, \acute{a})\rightarrow\mathcal{G}_n} r^{a}_{k,n}>\Sigma_{\mathcal{G}_n\diagdown (\acute{k}, \acute{a})} r^{a}_{k,n}$ can not be realized. Similarly, we can derive that the stable pair $(k,a)$ can not remove from $\mathcal{G}_n$ and there is no matching $\acute{\mu}$ which is preferred to stable $\mu^*$.   It is noticeable that  at eCA Algorithm, only associated user $k\in \mathcal{K}_a$ is assumed in the matching process for subcarriers in set $ \mathcal{N}_a$.
                      Also, only subcarriers and FBSs that their capacity queues are empty
                       regarding to $ q_\text{max} $ and $ \hat{N}_a $ can be taken into acount in iteration $ l $. Due to the finite number of FBSs and subcarriers, the sets of the preference relations of FUEs, transmission nodes and subcarriers are also finite. Furthermore, the convergence of the swap-matching method follows from some considerations as: the number of possible swaps is finite due to the fact that each FUEs can reach a limited number of FBSs in its vicinity and also the subset of these swaps may be preferred.
       \end{proof}
        \begin{lemma} \label{lemma1}
        Matching sub-problems in Algorithm \ref{GPsubcarrier} converge to a lacal optimal solution.
        \end{lemma} 
        \begin{proof}
        See Appendix
        \end{proof}
                \begin{lemma}\label{lemma2} 
                Beamforming sub-problems in Algorithm \ref{GPsubcarrier} converge to a lacal optimal solution.
               \end{lemma} 
               \begin{proof}
                       See Appendix
                       \end{proof}
\begin{theorem} \label{theorem3}
Algorithm \ref{GPsubcarrier} is guaranteed to converge.
\end{theorem} 
\begin{proof}
Algorithm \ref{GPsubcarrier} has three sub-problems. Based on Lemma \ref{lemma1}, Lemma \ref{lemma2}, and Theorem \ref{theorem1},  each of them  converges to its maximum local solution, and also due to the finite number of variables and  
feasible sets, the overal Algorithm \ref{GPsubcarrier} garanteed to be converged.
\end{proof}
\section{COMPUTATIONAL COMPLEXITY}\label{seccomputation}
In this section, we discuss the computational complexity of the proposed resource allocation approaches. The applied matching algorithm  to solve \eqref{eq:2} includes three steps: Cooperative node selection through CTNSA, determining the subcarrier allocation and  beamforming. In CTNSA algorithm, we assume that the maximum allowable number of cooperative nodes are used. Therefore, the worst-case complexity of transmission node selection can be determined by $\Sigma_{a=1:F_{\text{max}}}K_{CS}(l)(FT_f-a)$ where $K_{CS}(l)$ defines the number of FUEs can join into the CS phase at iteration $l$. The complexity of eCA phase is determined by $K_{CA}(l)N_cFT_f$ where $K_{CA}(l)$ is the number of FUEs join into the eCA phase at iteration $l$. It is significant
 that complexities of the worst-case and stochastic approaches are same at both CS and eCA phases while in the beamforming phase they are different. In stochastic case, CVX selects the SDP method for power allocation. This method can be handled by the interrier point method with a worst-case complexity of $ O(\text{Max}\{m,n\}^4n^{\frac{1}{2}}\times\text{log}(\frac{1}{\varrho})) $ with the solution accuracy $ \varrho $ where $ m $ is the number of constraints and $ n $ is the problem size\cite{sdr2} and \cite{complexity}.
In power allocation $ m $ is $ N_c \times \big{(} 2K^2+2K+1\big{)}$ and $ n $ equals $ F\times T_f $. Computational complexity of the proposed SCA approach in subcarrier allocation stage is similar to the beamforming problem plus $ N_c FT_f \times\big{(} K+2 \big{)} +KN_c+FT_f$.
 In the worst-case, all of the subcarrier and power allocation stages are similar to the stochastic method, but the number of constraints in power allocations is $ m = N_c \times \big{(}  K^2+ K+1\big{)}$ and $ n $ equals $ F\times T_f $. 

 
 \section{SIMULATION RESULTS}\label{secsimulation}

\begin{figure}
	\centering
	\includegraphics[width=9 cm, height =6 cm]{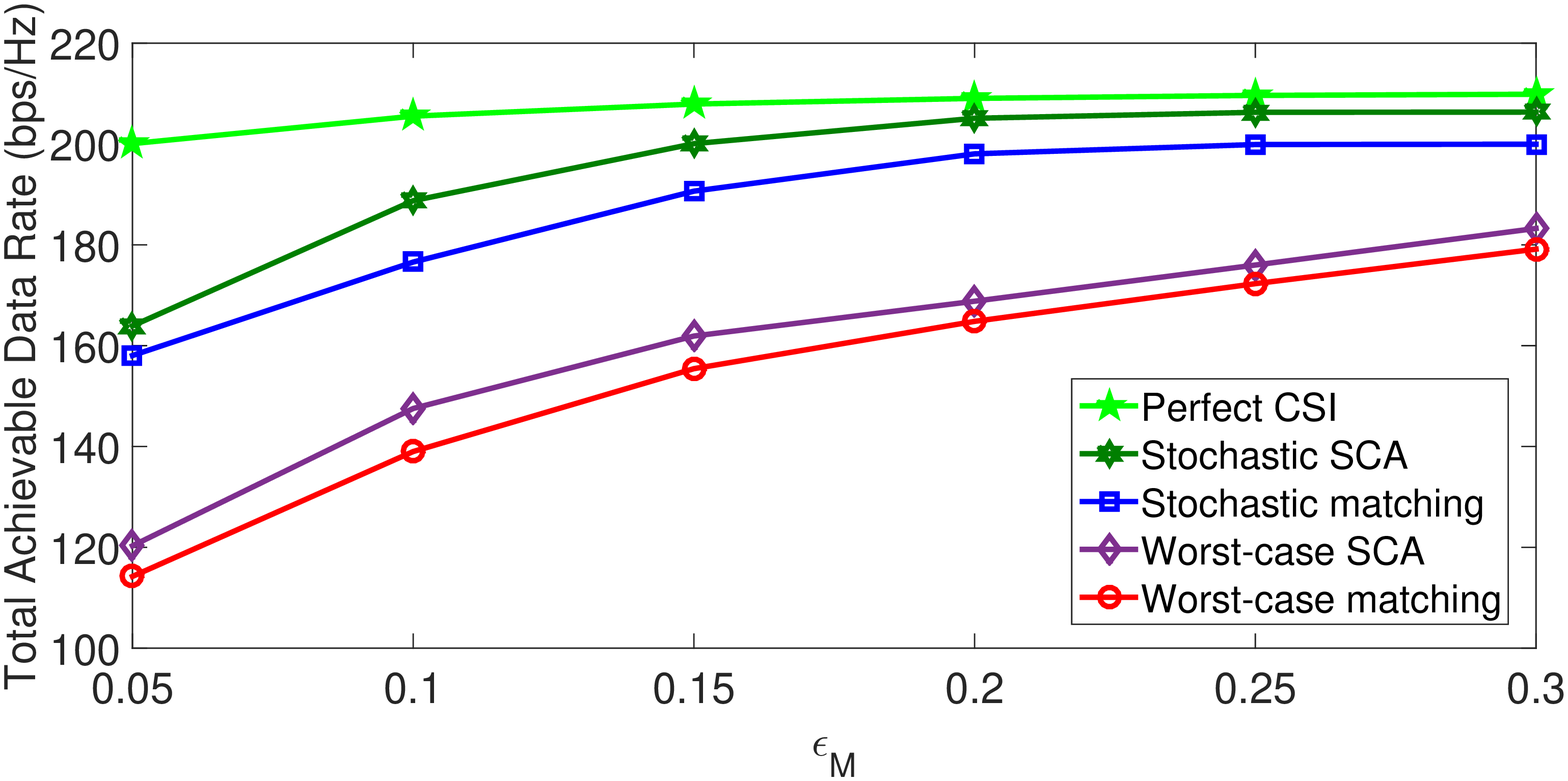}
	\caption{Comparison performance between the employed robustness methods for achievable sum rate versus $ \epsilon_M $.}
	\label{fig:bernvsworstcase1l}
\end{figure}

\begin{figure}
	\centering
	\includegraphics[width=9 cm, height =6 cm]{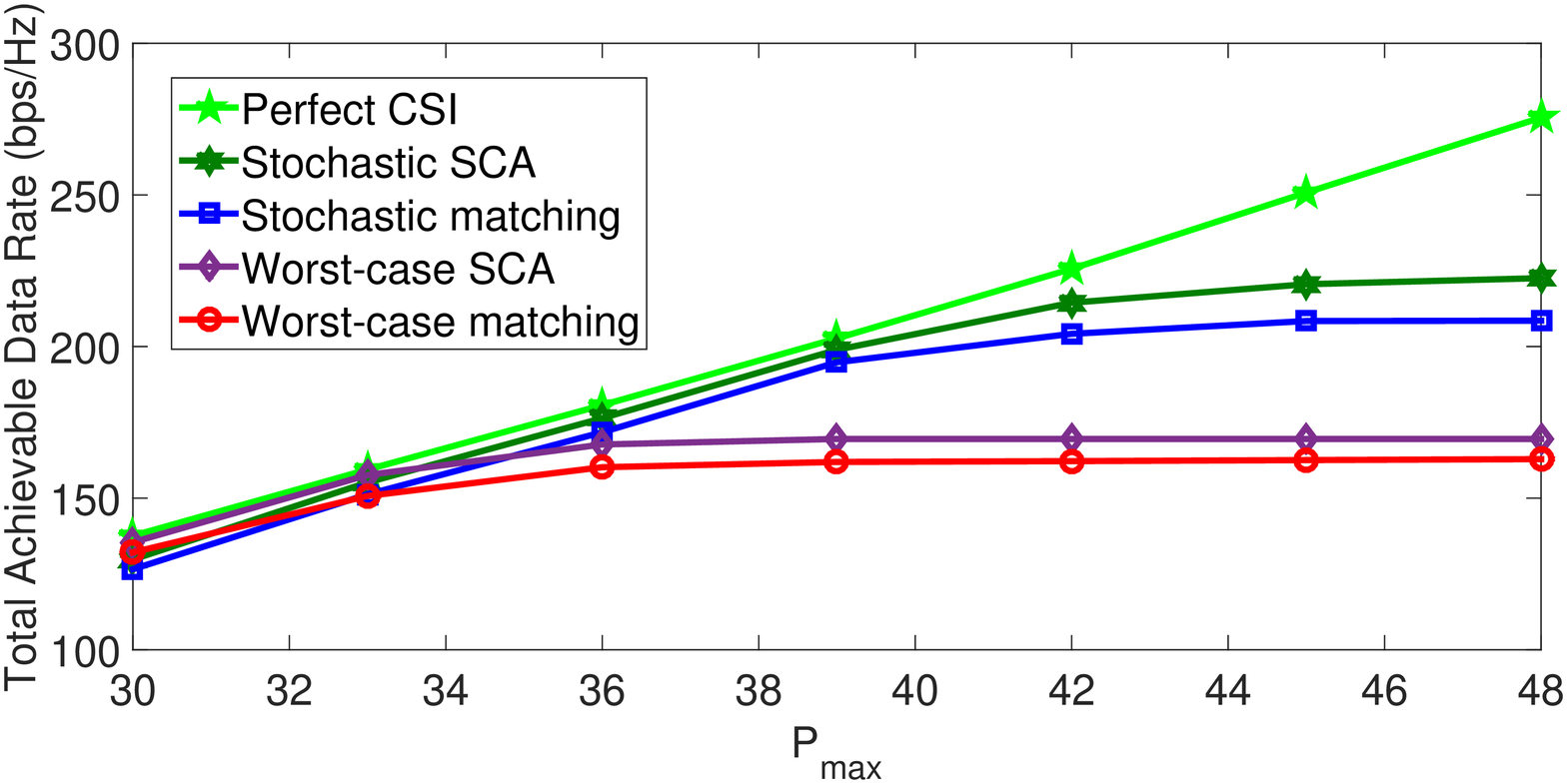}
	\caption{Comparison performance between the employed robustness methods for achievable sum rate versus different  $ P_\text{max} $.}
	\label{fig:bernvsworstcase1r}
\end{figure}
 
 In this section, we evaluate the performance and complexity of the proposed algorithms. 
  \subsection{Simulation setup}
We consider a HetNet with 3 FBSs where  the $f^{th}$ FBS is equipped by 2 antennas and $F_\text{max}$ as the descriptive parameter for the maximum allowable number of cooperative femtocells and receivers on the same subcarrier is preset 3 where the number of subcarriers is $ N_C=10 $. The order of MISO is specified, i.e., $T_f = 2$ and $T_m = 8$.  We consider that $ \sigma_{k,n}^{2}=10^{-4} $ for all results and $ \epsilon_M=0.2 $. For simplicity, We set the variance of channel coefficients and errors as  $ \delta^2_{\textbf{h},\mathcal{F}_{k,n}}=1 $, $ \delta^2_{\textbf{h},\mathcal{FM}_{k,n}}=\delta^2_{\textbf{h},\mathcal{MF}_{n}}=0.05 $ and  $ \delta^2_{\textbf{e},\mathcal{F}_{k,n}}= \delta^2_{\textbf{e},\mathcal{FM}_{k,n}}=\delta^2_{\textbf{e},\mathcal{MF}_{n}}=0.001$ which are equals for all of the users. Robustness parameters assumed $ \zeta=\kappa=0.05 $, $ \eta=0.2 $ and $ \alpha=\beta=0.2 $.
 Moreover, the target of achievable data rate is assumed $ R_k=0.3 \:(\text{bps/Hz}) $ and the maximum power budget of the FBSs sets $ 40 \:\text{dBm} $. For matching algorithms, we select $  \varUpsilon_{CS}=\varUpsilon_{CA}=100 $, $c_{\mathcal{FM}}^n=5$ and $c_i^n=0.2$. 
 \subsection{Simulation results}
  Fig. \ref{fig:bernvsworstcase1l} shows the achievable sum rate versus $ \epsilon_M $.  As expected, applying
  robustness methods decreases achievable data rate as a robustness cost which is rational against perfect CSI conditions. 
 We employ the loosely bounded robust solution as a worst-case method to ensure  the performance of the network specially in critical conditions. It is clear that by increasing the $ \epsilon_M $, the sum rate is increased in all of the employed methods. Further, the achievable sum rate in the Bernstein approach is more than the other one and the performance of devised Bernstein method is near the ideal case. 
 Moreover, Fig. \ref{fig:bernvsworstcase1l} shows the proposed SCA method has more performance in comparison with the proposed matching method.
\begin{figure}
	\centering
	\includegraphics[width=9 cm, height =6 cm]{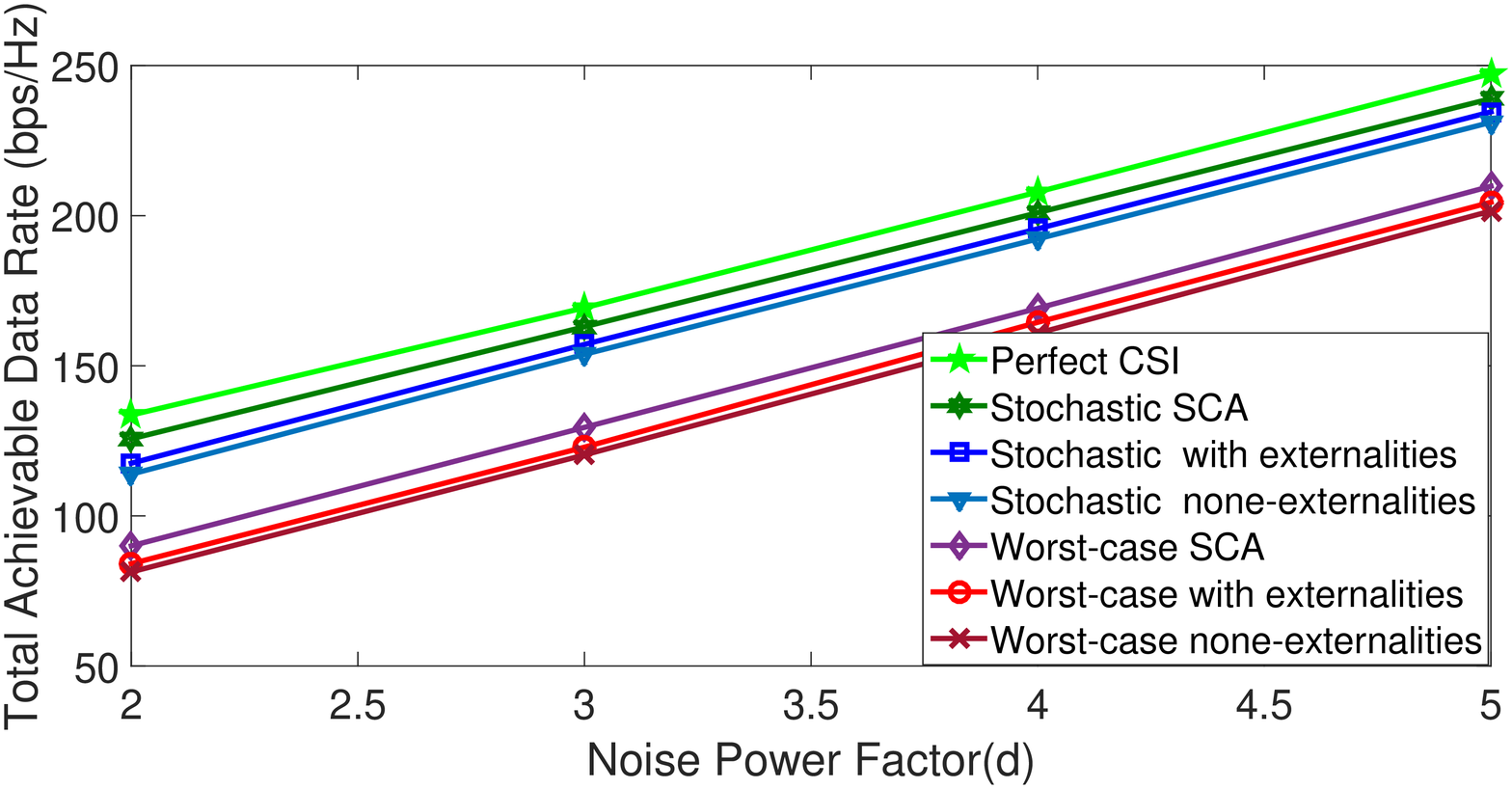}
	\caption{Comparison performance for achievable sum rate versus different noise powers.}
	\label{fig:fmaxl}
\end{figure}

Fig. \ref{fig:bernvsworstcase1r} demonstrates the total achievable data rate versus the different $ P_\text{max} $.
By increasing the power budget of each FBS, the more sum rate can be achieved, but the incremental process is limited in higher values of $ P_\text{max} $. Specially for worst-case approach, this limitation is strict. 
 The performance of the stochastic approaches can achieve to high sum rate with low robustness cost as power budget.
Fig. \ref{fig:fmaxl}  shows the performance of the devised methods for different noise power factors (d) where $ \sigma^2=10^{-d} $. The throughput is increased in the low noise power situations. Moreover, assuming the swap-matching algorithm increases the performance of the matching method.

 To compare sensitivity of the proposed methods, we compute achievable sum rate versus variable robustness parameters. We consider that the robustness parameters of deterministic and probabilistic approach, i.e.,
 $ \eta $ and $ \alpha $ are variable, i.e.,
 each point of horizontal axis in Fig. \ref{betal} 
 represents $ \alpha $ and $ \eta $ so that
  $ \alpha = \eta $.
  The figure shows that although
   increasing $ \alpha $ increases the achievable sum rate, this increment is not too much against another one. Actually, variation of $ \eta $ has a great impact on results and the more $ \eta $ increases, the more sum rate are attained in worst-case. Fig. \ref{betal} illustrates that the performance of the worst-case method is related to the allowable error bounds. Therefore, worst-case is an unreliable
   and sensitive method. The shape of the ellipsoid that we choose in the worst-case method has a deep influence on results. Little $ \eta $ gives more achievable data rate though the model of CSI uncertainty set may be inaccurate and unreliable. Moreover, in Fig. \ref{betal},  comparison of the performance between PD-NOMA and OFDMA based networks is presented.  
    For the feasible initialized state of Fig. \ref{betar}, the beamforming vectors are determined with the introduced worst-case problem in  \cite{atefeh2016}.
     Fig. \ref{betar}  shows the number of required iterations to
   achieve convergence for $ \sigma_{k,n}^{2}=10^{-2} $. Moreover, comparison between the stochastic and worst-case approaches shows that the computational complexity of the stochastic Bernstein approach is more than worst-case method.

\begin{figure}
	\centering
	\includegraphics[width=9 cm, height =6 cm]{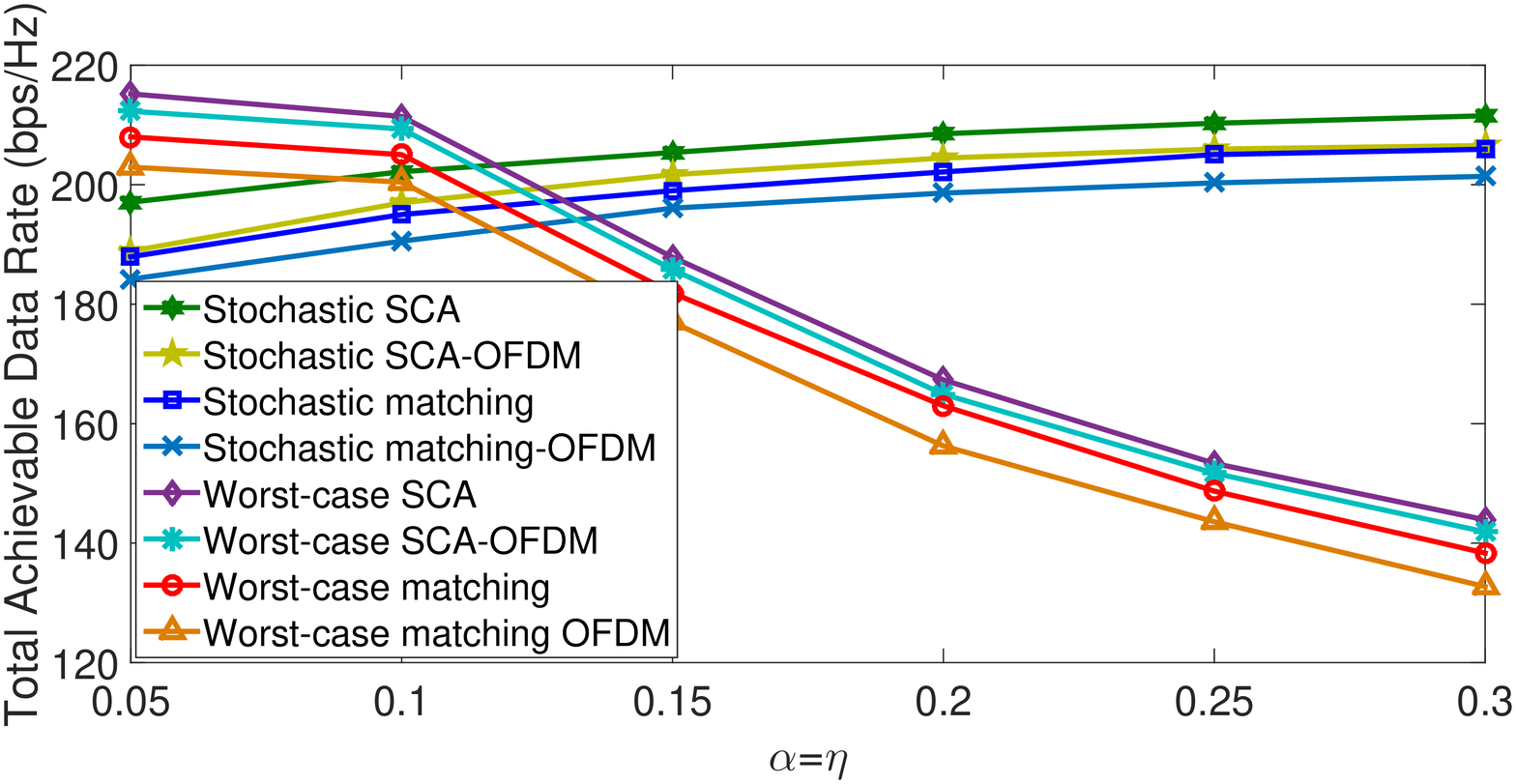}
	\caption{Comparison performance between the employed robustness methods for achievable sum rate versus different robustness factors.}
	\label{betal}
\end{figure}

\begin{figure}
	\centering
	\includegraphics[width=9 cm, height =6 cm]{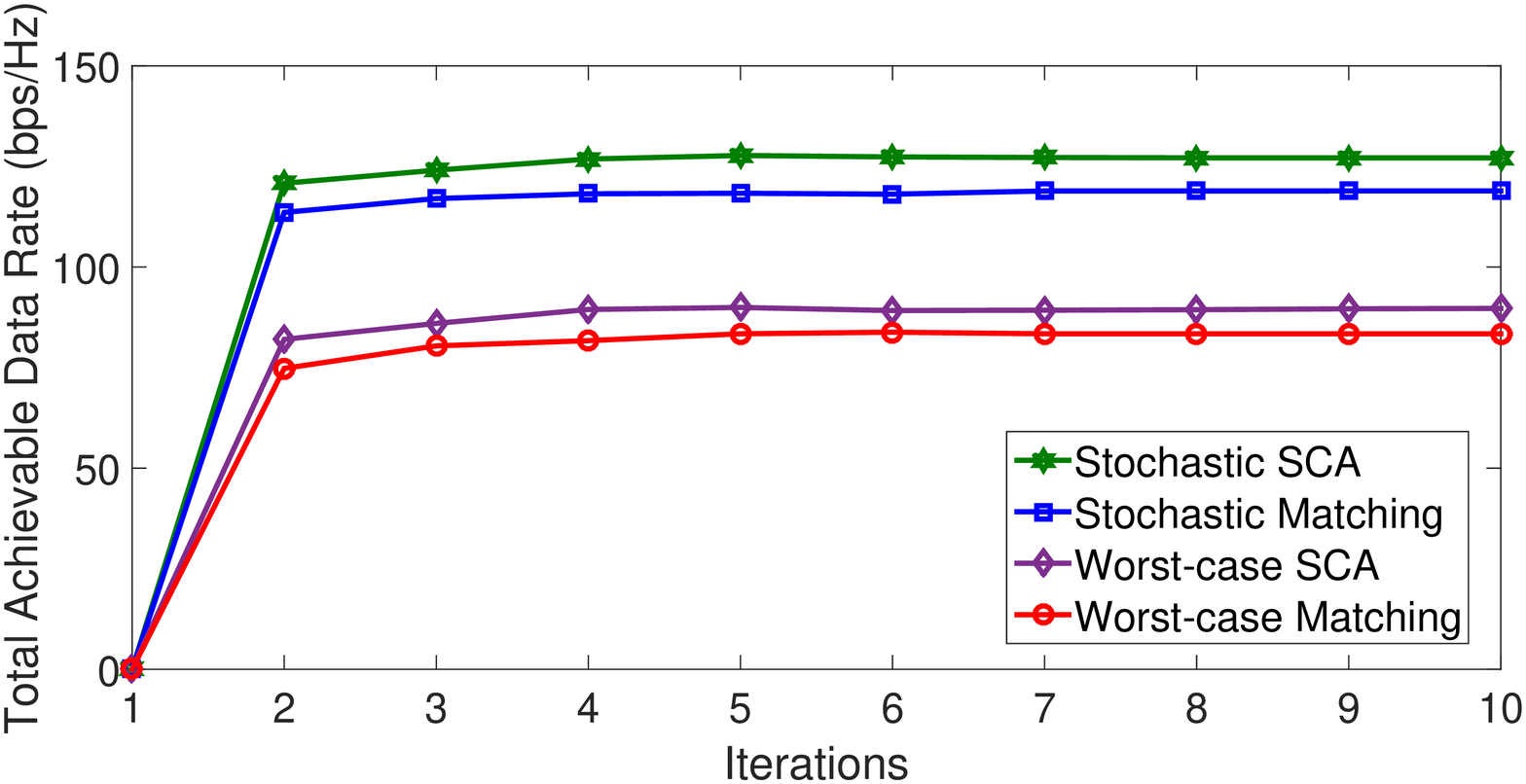}
	\caption{ Convergence of
	the proposed solutions.}
	\label{betar}
\end{figure}
\section{Conclusion}\label{secconclusion}
In this paper, a 5G  network which works based on hybrid category of CoMP technology was considered. We proposed some new advanced methods for the cooperative nodes association and subcarrier allocation problem based on matching game with externalities and SCA for this network. Further, due to the uncertainty of CSIT,  two robustness schemes were considered, and based on them, two methods were proposed  to ensure user's satisfaction in both worst-case and stochastic cases. Because of the high signaling, we devised the worst-case method which is known as a low complexity method. Moreover, we proposed a probabilistic methods based on Bernstein inequality. As expected, by increasing the  users demand, more resources must be allocated to satisfied them. 
 To achieve high capacity, the Bernstein approach needs less robustness cost as transmit power compared with the  worst-case methods.
We also recognized that the computational complexity of the  worst-case method is lower than that of the probabilistic approach. Therefore, we can choose either the worst-case or B‌ernstein approach based on our preference and facilities.  As a future work, we will study the beamforming problem for MBS to increase the achievable data rate for MUEs.
\appendix

\begin{proof}[Proof of the Lemma \ref{lemma1}]
 In the CTNSA Algorithm, the stable matching $ \mu_{CS} $ based on new stable group in each iteration and the reject/accept operation guarantee that in $ (l+1)^{th} $ time 
$ \Sigma_{a\in \mathcal{A}}\Sigma_{k\in \mathcal{K}}\varphi_{CS}^a(k)^{(l+1)}\geq \Sigma_{a\in \mathcal{A}}\Sigma_{k\in \mathcal{K}}\varphi_{CS}^a(k)^{(l)} $ which shows matching function is a non-decreasing function. To express more details, we remark that we employ equations \eqref{ukcs} and \eqref{uacs} instead of equations \eqref{eq:2b} and \eqref{eq:2c}.  The utility functions $\varphi_{CS}^k(a) $ and $ \varphi_{CS}^a(k) $ are compatible with increasing of the
$ \frac{|\bar{h}_{\mathcal{F}_{k,n}}^a\hat{w}_{k,n}^a|^2}{2^{R_k}-1} $ for fixed subcarrier and beamforming design.
Also, by introducing $ \varTheta_{k,a} $ in utility function $  \varphi_{CS}^a(k)$, it can be guaranteed that the interferences of the SINR function can be minimized.  In accordance with these utilities,  the SINR function and also the total data rate are increased too.
Therefore, the proposed CTNSA algorithm converges to a local optimal solution for problem \eqref{CTNSA}. As the CTNSA Algorithm, \eqref{ukcs1} has a direct effect on the data rate function and introducing \eqref{uacs11}  decreases the total interferences which is compatible with the overal throughput of the network, the eCA Algorithm converges to its maximal value which is a local optimal solution.
\end{proof}

\begin{proof}[Proof of the Lemma \ref{lemma2}]
          We aim to prove that each of the proposed robust approaches with the DC approximation converges to its local optimum. The objective functions of \eqref{eq:40} and \eqref{eq:finberns1} are approximated by the DC approach where $ g_{k,n} $ is concave and $ \nabla g_{k,n}(\textbf{W}_{k,n}) $ as its gradient is also
      its super-gradient\cite{rconv1}. Hence, in iteration $ l $, we have
       \begin{align}\label{conv1}
           & g_{k,n}(\textbf{W}^{[l]}_{k,n})\leq g_{k,n}(\textbf{W}^{[l-1]}_{k,n})+\nonumber \\ & \langle \nabla g_{k,n}(\textbf{W}^{[l-1]}_{k,n}),(\textbf{W}^{[l]}_{k,n}-\textbf{W}^{[l-1]}_{k,n})\rangle.
          \end{align}
        As the objective is equal to \eqref{dfg}, we have
         \begin{align}\label{conv2}
      & f_{k,n}(\textbf{W}^{[l]}_{k,n})-g_{k,n}(\textbf{W}^{[l]}_{k,n})\geq \nonumber \\ & f_{k,n}(\textbf{W}^{[l]}_{k,n})-[g_{k,n}(\textbf{W}^{[l-1]}_{k,n})+\nonumber \\ & \langle \nabla g_{k,n}(\textbf{W}^{[l-1]}_{k,n}),(\textbf{W}^{[l]}_{k,n}-\textbf{W}^{[l-1]}_{k,n})\rangle]\geq \nonumber \\
         & f_{k,n}(\textbf{W}^{[l-1]}_{k,n})-g_{k,n}(\textbf{W}^{[l-1]}_{k,n}).
         \end{align}
        The incremental process for the objective of subcarrier allocation in the proposed SCA method can be guaranteed as \eqref{conv2}.  Due to \eqref{conv2}, after iteration $ l $, the objectives of \eqref{eq:40} and \eqref{eq:finberns1} improve against previous
      solution or is almost equal. As mentioned in \cite{rconv2}, the SCA approach with the DC
     approximation is guaranteed to converge to a local optimum. Consequently, we have 
           \begin{align}\label{conv3}
           & r_{k,n}(\boldsymbol{\varrho}^{[l-1]}_{k,n},\textbf{W}^{[l-1]}_{k,n})\leq   r_{k,n}(\boldsymbol{\varrho}^{[l-1]}_{k,n},\textbf{W}^{[l]}_{k,n})\leq 
              r_{k,n}(\boldsymbol{\varrho}^{[l]}_{k,n},\textbf{W}^{[l]}_{k,n}).
         \end{align}
            Based on \eqref{conv2}, the proposed solution at the end of each iteration is better than the previous iteration and for a finite set of transmit powers and channel gains, the optimal achievable sum rate is bounded abov. Thus, the procedure of improving the solutions always converges.
\end{proof}

\clearpage
\end{document}